\documentclass[conference]{IEEEtran}
\usepackage{caption}
\usepackage{algorithm}
\usepackage{algorithmic}
\usepackage{graphicx}
\usepackage{enumerate}
\usepackage{amsthm}
\usepackage{epstopdf}
\usepackage{amsmath,amssymb,epsfig,graphicx,slidesec,colortbl}
\usepackage[absolute,overlay]{textpos}
\usepackage{multimedia,algorithm,algorithmic,multirow,url}
\usepackage{cite}
\usepackage{subfig}
\usepackage{color}
\usepackage{url}
\usepackage{enumitem}

\newtheorem{Def}{Definition}
\newtheorem{The}{Theorem}
\newtheorem{Pro}{Proposition}
\newtheorem{Cor}{Corollary}

\newtheorem{Lem}{Lemma}
\newtheorem{Obs}{Observation}

\usepackage{geometry}
\IEEEoverridecommandlockouts
\begin{document}
	%
	
	% paper title
	% Titles are generally capitalized except for words such as a, an, and, as,
	% at, but, by, for, in, nor, of, on, or, the, to and up, which are usually
	% not capitalized unless they are the first or last word of the title.
	% Linebreaks \\ can be used within to get better formatting as desired.
	% Do not put math or special symbols in the title.
	\title{Strategic Information Revelation in \\Crowdsourcing Systems Without Verification }
	\author{Chao Huang, Haoran Yu, Jianwei Huang, and Randall A. Berry
	
		\thanks{C. Huang is with the Department of Information Engineering, the Chinese University of Hong Kong, Hong Kong; Email: hc017@ie.cuhk.edu.hk. H. Yu is with the School of Computer Science, Beijing Institute of Technology, China; Email: yhrhawk@gmail.com.  J. Huang is with the School of Science and Engineering, The Chinese University of Hong Kong, Shenzhen, and the Shenzhen Institute of Artificial Intelligence and Robotics for Society (corresponding author, e-mail: jianweihuang@cuhk.edu.cn). R. A. Berry is with the Department of Electrical and Computer Engineering, Northwestern University, USA; Email: rberry@northwestern.edu. This work is supported by the Shenzhen Institute of Artificial Intelligence and Robotics for Society, and the Presidential Fund from the Chinese University of Hong Kong, Shenzhen. The work of Haoran Yu is supported by Beijing Institute of Technology Research Fund Program for Young Scholars.
		}
	}
	%\vspace{-7mm}
	
	% 
	%\author{\IEEEauthorblockN{Chao Huang\IEEEauthorrefmark{1}, Haoran Yu\IEEEauthorrefmark{2}, Jianwei Huang\IEEEauthorrefmark{1}, Randall Berry\IEEEauthorrefmark{2}}
	%	\IEEEauthorblockA{\IEEEauthorrefmark{1}Department of Information Engineering, The Chinese University of Hong Kong}
	%	\IEEEauthorblockA{\IEEEauthorrefmark{2}Department of Electrical Engineering and Computer Science, Northwestern University }
	
	%\thanks{The work is supported by the Theme-based Research Scheme (Project No. T23-407/13-N) from the Research Grants Council of the Hong Kong Special Administrative Region, China, by  a grant from the Vice-Chancellor's One-off Discretionary Fund of The Chinese University of Hong Kong (Project No. VCF2014016), by the NSF grants CCF-1442726 and ECCS-1509536, and in part by  the University of Washington Clean Energy Institute.}
	\vspace{-10mm}
	\maketitle
	
	% As a general rule, do not put math, special symbols or citations
	% in the abstract
	
	\begin{abstract}
		
			%Crowdsourcing utilizes distributed  workers for finishing tasks (e.g., image labeling) and providing high-quality and truthful solutions. \textit{Information elicitation without verification (IEWV)} is 
		We study a crowdsourcing problem where the platform aims to incentivize distributed workers to provide high-quality and truthful solutions without the ability to verify the solutions.
		While most prior work assumes that the platform and workers have symmetric information, we study an asymmetric information scenario where the platform has informational advantages. Specifically, the platform knows more information regarding workers' average solution accuracy, and can strategically reveal such information to workers. Workers will utilize the announced information to determine the likelihood that they obtain a reward if exerting effort on the task.
		%The platform's strategic information revelation is a challenging non-convex problem, which we solve via exploiting its special structure.!!!!!!!!!!!!!!!!!!!!!!!!!!!!!!!!!!!!
		We study two types of workers: (1) \textit{naive} workers who fully trust the announcement, and (2) \textit{strategic} workers who update prior belief based on the announcement. For naive workers, we show that the platform should always announce a high average accuracy to maximize its payoff. 
		However, this is not always optimal for strategic workers, as it may reduce the credibility of the platform's announcement and hence reduce the platform's payoff. Interestingly, the platform may have an incentive to even announce an average accuracy lower than the actual value when facing strategic workers. Another counter-intuitive result is that the platform's payoff may decrease in the number of high-accuracy workers. 
		%A larger number of high-accuracy workers brings marginally decreasing benefits, but the cost required to incentivize them may grow significantly.
		
		%We adopt the $\epsilon$-equilibrium concept, under which a worker cannot achieve a payoff increase larger than $\epsilon$ by deviating from the equilibrium. 
		%Under our proposed learning mechanism, a worker's payoff increase by deviating from equilibrium is bounded if his accuracy misreport (i.e., the difference between the reported accuracy and the actual accuracy) is bounded. Furthermore, the payoff increase from deviation approaches zero in the long run, which induces workers to truthfully report their  solution accuracy levels at an equilibrium, in which the platform asymptotically achieves zero regret. 
		%The regret indicates the difference between the empirical accuracy distribution (collected from workers' accuracy reports) and the actual accuracy distribution of the workers.  
		
	\end{abstract}
	
	% no keywords

	% For peer review papers, you can put extra information on the cover
	% page as needed:
	% \ifCLASSOPTIONpeerreview
	% \begin{center} \bfseries EDICS Category: 3-BBND \end{center}
	% \fi
	%
	% For peerreview papers, this IEEEtran command inserts a page break and
	% creates the second title. It will be ignored for other modes.
	\IEEEpeerreviewmaketitle

	\section{Introduction}
	\subsection{Motivations}
	The rapid growth of the Internet has enabled crowdsourcing of various online tasks \cite{akimoto2019crowdsourced,tian2018mobicrowd}.  Through appropriately assigning decomposed simple tasks to workers and effectively aggregating workers' solutions, a crowdsourcing platform can manage to solve the original complex problem \cite{yin2017task}. 
	For example, Waze invites workers to report local traffic information and uses the reports to recommend driving routes \cite{waze}.  Amazon Mechanical Turk recruits workers to do image labeling, and uses the collected labels to train machine learning models \cite{Mturk}. 
	%In IdeaConnection, a popular crowdsourcing website, users are encouraged to submit innovative solutions to technical challenges such as accurate measurement of oceanographic data \cite{ideaconnection}. 

	Even for a simple task, obtaining a high-quality solution requires a worker to exert enough effort. A platform needs to provide proper incentives
	to motivate this \cite{jin2019data}. The design of incentives is particularly difficult when the
	platform cannot access the ground truth to verify the workers’ reported solutions. For example, the ground truth
	can be costly to obtain (e.g., it is costly for Waze to judge the accuracy
of the mobile workers' reported traffic information, as it requires
centralized managed sensors to validate all the submitted reports). This type of challenging crowdsourcing problem is known as information elicitation without verification
	(IEWV) \cite{kong2018water}.
	
	%To encourage high-quality solutions from the crowdsourced workers, a platform needs to carefully design the reward mechanism. This is challenging, especially when the platform cannot verify the correctness or quality of workers' solutions \cite{kong2016putting,chaoGC19}. This may be because it is costly and time-consuming to obtain the ground truth. For example, in Waze, traffic jam information reported by the users may be difficult to verify in real-time, and the platform can only reward users based on their reported data. In OpenReview, it is also hard to verify the quality of reviews due to the intrinsic complexity of academic review and reviewer anonymity.
	%%\footnote{Crowd sensors can be smart phones and wearable devices that are capable of sensing information such as location, light, and movement \cite{crowdsensing}.}
	%When there is a lack of verification for the workers' solutions, the crowdsourcing problem is known as \emph{information elicitation without verification} (IEWV) \cite{waggoner2014output}.

	Most IEWV literature (e.g., \cite{waggoner2014output,miller2005eliciting,li2019service,frongillo2015elicitation,shnayder2016informed}) studied the problem as a game with symmetric information. The common assumption is that both the workers and the platform have the same information regarding the environment, e.g., worker capabilities.  
	In many crowdsourcing platforms, however, information regarding the worker characteristics is  \textit{asymmetric} between the platform and workers. Usually, the platform has more information regarding worker characteristics through market research and past experiences.  For example, in Amazon Mechanical Turk, each worker's historical performance is known by the platform, but not by the other workers \cite{Mturk}.  
	%With such informational privilege, the platform can strategically announce the worker information to induce desired worker behaviour and maximize its own payoff.
	In this paper,
	we consider information asymmetry between the platform and the workers.  As will be seen, such information asymmetry complicates the analysis of both the workers' behaviors and the platform's incentive design.
	
	\begin{figure*}
		\centering
		\vspace{-5mm}
		\includegraphics[width=5.8in]{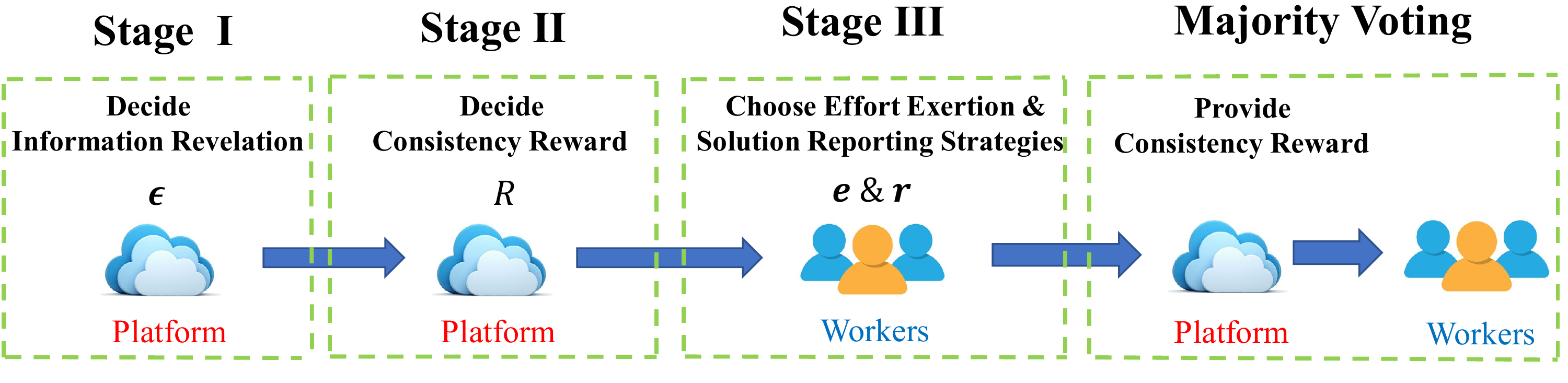}
		\caption{Interactions between platform and workers.}
		\label{SystemModel}
		\vspace{-3mm}
	\end{figure*}

	We apply the widely adopted  \textit{majority voting} scheme to solve the IEWV problem \cite{8387487,liu2016learning,chaoGC19}. Specifically, a worker obtains a \textit{consistency reward} if his solution matches the majority solution from the other workers.  Besides the reward design under majority voting, we assume that the platform has one additional decision: information revelation.
	%\footnote{Some economic literature studied sellers who reveal product information to buyers and used \emph{information transmission} to denote the sellers' decisions \cite{crawford1982strategic,argenziano2016strategic}. In this paper, we use ``information revelation'' and ``information transmission'' interchangeably.}
	 Specifically, we consider a mix of high and low accuracy workers, where the platform knows the number of each type  (but the workers do not know). 
%	 We want to understand how the platform will take advantage of this information asymmetry, and it raises two important questions:
%	 \begin{itemize}
%	 	\item \textit{Does the platform have an incentive to reveal the information to the workers?}
%	 	\item  \textit{If so, does the platform have an incentive to even manipulate the revealed information?}
%	 	\end{itemize}
	 \textbf{We are particularly interested in studying whether the platform has an incentive to reveal this information to workers. Furthermore, we also want to study whether the platform may manipulate the revealed information}, as truthful revelation may not benefit the platform or the workers \cite{crawford1982strategic,argenziano2016strategic,kamenica2011bayesian,kurschilgen2019communication}. 
	%In our paper, we model the platform's information revelation/transmission as its decision . 
	%However, the coupling between the information revelation and the reward design renders the problem non-convex and challenging.  By exploiting the problem's special structure and analyzing the optimal decisions, 
	By analyzing the platform's optimal information revelation decisions, we can  understand how the platform takes advantage of this information asymmetry. 
 
	%Some prior work also studied information transmission problems under information asymmetry [...][...][...]. However, the difference here is that ...

	%The mechanism design considering information asymmetry is very challenging in our context. This is because neither the platform can verify the workers' task solutions nor the workers can check the authenticity of the platform's information releasing. Moreover,  the platform's information releasing decision affects the workers' beliefs and hence further couples with its reward design.

	We model the interactions between the platform and the workers as follows (illustrated in Fig. \ref{SystemModel}): 
	\begin{enumerate}
		\item \textit{Stage I}: The platform decides the information revelation strategy. 
		\item \textit{Stage II}: The platform decides the consistency reward under the majority voting scheme.
		\item \textit{Stage III}: The workers choose whether to exert effort finishing the task and whether to truthfully report solutions.
	\end{enumerate}

	After workers finish the task and report their solutions in Stage III,  majority voting is implemented. Specifically, the platform collects the workers' reported solutions and allocates consistency rewards to the workers whose solutions are consistent with the majority. Note that the majority voting's results are directly determined by the previous three stages, and hence we do not treat it as a separate stage.
	
		We study two types of workers: (1) \textit{naive} workers who fully believe  in the platform's announcement, and (2) \textit{strategic} workers who update their prior belief based on the announced information.  The consideration of naive worker case serves as a benchmark and it can model the scenario where workers are confident in the platform's announced information (especially for those platforms with good reputations). It can also model the scenario where workers have limited reasoning capabilities to deduce the authenticity of the announcement  \cite{shao2019multimedia}. The consideration of strategic worker case, however, leads to more intriguing results.  Such a consideration fits the scenario where workers do not trust the platform and  where workers are strategic and have high reasoning capabilities  \cite{chaoWiOpt20}. 
	As will be shown, the platform's optimal information revelation strategies are very different when facing these two types of workers. 
	\subsection{Key Contributions}
	
	The main contributions of this paper are as follows.
	\begin{itemize}
		\item \emph{Studying strategic information revelation for the IEWV problem}: 
		To the best of our knowledge, this is the first analytical work that studies strategic revelation of asymmetric information for the IEWV problem. 
		The platform's information revelation is a challenging non-convex problem, yet we can exploit its special structure to characterize the optimal solutions' properties.
		%Our work opens up an exciting direction for future study and bears interesting and important practical implications.

		\item \emph{Characterizing workers' equilibrium strategies}: We show that there are multiple equilibria among the workers. Under appropriate information revelation and reward design, it is a Pareto-dominant equilibrium for all the workers to exert effort and truthfully report their solutions.
		
		\item \emph{Characterizing platform's information revelation strategy}: We show that for naive workers, the platform should always announce a high average worker accuracy, independent of the actual information. 	However, this is not always optimal when facing strategic workers, as it may reduce the credibility of the platform's announcement and hence reduce the platform's payoff. When facing strategic workers, the platform may even have an incentive to announce an average accuracy lower than the actual value.
		\item \emph{Performance evaluation}: We evaluate our mechanism via extensive numerical experiments. We show that the platform payoff increases in the workers' prior belief regarding the number of high-accuracy workers. Surprisingly, the platform payoff may decrease in the number of high-accuracy workers. A larger number of these workers brings marginally decreasing benefits, but the cost required to incentivize them may grow significantly.
		
	\end{itemize}
	
	The remaining of the paper is organized as follows. In Section II, we review the related work.  In Section III, we introduce the model. In Section IV, we provide complete analytical solutions to the model. We show numerical results in Section V and conclude in Section VI.

	%The rest of this paper is organized as follows. In \textcolor{black}{Section} \ref{model}, we introduce the model. In Sections \ref{worker} and  \ref{platform}, we analyze the workers' and the platform's optimal strategies, respectively. We show numerical results  in  \textcolor{black}{Section} \ref{numerical} and conclude the paper in \textcolor{black}{Section} \ref{conclusion}.

	\section{Related Work}
	Our work studies strategic information revelation for IEWV. Hence, we review the related work from two aspects, i.e., the IEWV problem and strategic information revelation.
	
	\subsection{Information Elicitation Without Verification (IEWV)}
	IEWV  is a canonical crowdsourcing problem where the platform cannot verify the workers' reported solutions (e.g., \cite{witkowski2015robust,prelec2004bayesian,waggoner2014output,chaoGS19}). 
	The key focus in IEWV is to design proper incentive mechanisms to encourage high-quality (and truthful) solutions. 
	Miller \textit{et al.} in \cite{miller2005eliciting} and Prelec \textit{et al.} in \cite{prelec2004bayesian} proposed peer prediction mechanisms to elicit truthful solutions from the workers. Dasgupta \textit{et al.} in \cite{dasgupta2013crowdsourced} and Liu \textit{et al.} in \cite{liu2016learning} focused on designing mechanisms to induce workers' effort exertion, which leads to high-quality solutions. Huang  \textit{et al.} in \cite{chaoGC19,chaoWiOpt20} studied the impact of worker heterogeneity on the mechanism design for IEWV. However, these works did not consider information asymmetry and strategic information revelation between the platform and workers.  \textit{Our work is the first analytical study to characterize the impact of strategic  revelation of asymmetric information on the mechanism design for the IEWV problem.}
	
	\subsection{Strategic Information Revelation}
	Strategic information revelation investigates how the players in a game strategically reveal and process information to maximize their payoffs \cite{kurschilgen2019communication,frug2018strategic,hedlund2017bayesian}. Crawford \textit{et al.} in \cite{crawford1982strategic} explored the famous cheap talk problem and characterized the optimal structures of revealed information that maximizes the sender's payoff. Brocas \textit{et al.} in \cite{brocas2007influence} and Kamenica \textit{et al.} in \cite{kamenica2011bayesian} studied the persuasion game considering costly information acquisition and revelation. However, these works either made the strong assumption that the platform cannot lie  \cite{crawford1982strategic} or that the information revelation is the only platform decision \cite{kamenica2011bayesian}. \textit{Different from these works, we consider a general model where the platform can choose to be honest or lying. Further, we jointly consider the platform's strategic decisions on information revelation and reward design. Both the above considerations significantly complicate the analysis.}

	\section{Model}\label{model}
	
	%As previously noted, we model the interactions between the social platform and the workers as a two-stage Stackelberg game. We use backward induction to solve this game.
	In Section \ref{worker_problem}, we introduce the workers' decisions and payoffs. In Section \ref{platform_problem}, we introduce the platform's decisions and payoff, with an emphasis on the platform's information revelation strategy. 
	%In Section \ref{interaction}, we discuss the interactions between the workers and the platform.
	
	\subsection{Workers' Decisions and Payoffs}\label{worker_problem}
	
	In this subsection, we first introduce the task and workers, and then define each worker's strategies and payoff function.
	
	\subsubsection{Task and Workers}
	A crowdsourcing platform aims to obtain solutions to a task via a set  $\mathcal{N}=\{1, 2, \cdots, N\}$ of workers.
	%\footnote{Our analysis can be extended to the scenario where the worker set varies across time, as long as the worker distribution (i.e., accuracy distribution, which will be explained soon) and the total population $N$ stay unchanged.} 
	%The task is identical to workers in each time slot, but varies across different time slots.
	We consider a binary-solution task, e.g.,  judging whether the solution to an online math problem is \textit{Correct} or \textit{Wrong}. Let  $\mathcal{X}= \{1,-1\}$  denote the task's solution space, where $1$ means \textit{Correct} and $-1$ means \textit{Wrong},\footnote{Many crowdsourcing applications focus on binary tasks, e.g., image labeling and online content moderation, which draw extensive attention in literature \cite{chaoGC19,liu2016learning,chaoWiOpt20}. Besides, our model can be extended to the scenario where a task has more than two possible solutions, i.e., by decomposing a multi-solution task into several binary-solution tasks. For example,  eliciting opinions of  workers
		from three alternatives, i.e., good, bad or average on the
		quality evaluation of an online article is a three-solution task. We
		can decompose it into three binary-solution tasks, based on whether a worker's evaluation is good or not, bad or not, and average or not.}
	%For example, a task with three types $\mathcal{X}=\{A,B,C\}$ can be decomposed into three binary-type tasks with $\mathcal{X}_1=\{A, {\rm not} \; A\}$, $\mathcal{X}_2=\{B, {\rm not} \; B\}$, and $\mathcal{X}_3=\{C, {\rm not} \; C\}$. } 
	and $x\in \mathcal{X}$  is the task's true solution that the platform does not know. After completing the task, each worker  $i$ generates an estimated solution  $x_{i}^{\rm estimate} \in \mathcal{X}$, and he can report a value $x_{i}^{\rm report} \in \mathcal{X}$ to the platform that may or may not be the same as $x_{i}^{\rm estimate}$.

	\subsubsection{Worker Effort Exertion Strategy} Each worker can decide whether to exert effort doing the task, and the accuracy (i.e., quality) of his solution stochastically depends on his chosen effort level. Specifically, a worker can choose to either exert effort or not exert effort, and we use $e_{i} \in \{0,1\}$ to denote worker $i$'s effort level \cite{chaoWiOpt20,chaoGC19,chaoGS19}.  If worker $i$ does not exert effort,  i.e., $e_i=0$, he will generate the correct solution (which is the task's true solution) with probability $0.5$ at zero cost. Here, we assume that without exerting effort a worker has no information about the true solution, so the estimated solution is equally likely to be correct or wrong \cite{chaoGC19,chaoGS19}.\footnote{We can extend our analysis to the scenario where even without any effort a worker still has some information about the true solution. In this scenario, the solution would always be more accurate than random guessing \cite{liu2016learning,chaoWiOpt20}.} Exerting effort (i.e., $e_i=1$) improves a worker's solution accuracy at a cost $c\ge 0$, and he can generate the correct solution with probability $p_i\in (0.5, 1]$. 
	More specifically,
	\begin{equation}\label{effortmodel}
	\begin{aligned}
	\hspace{-3mm}Pr(x_{i}^{\rm estimate}=x)=
	\begin{cases}
	0.5, \quad &\text{if} \; e_{i}=0 \; (\text{with zero cost}),\\
	p_{i}, \quad &\text{if} \; e_{i}=1 \; (\text{with a cost} \; c \ge 0).
	\end{cases}
	%\vspace{-1.5mm}
	\end{aligned}
	\end{equation}
	In this paper, we consider heterogeneous workers in which there are $k$ out of $N$  workers with a high accuracy level $p_h$ and the remaining $N-k$ workers have a low accuracy level $p_l$, where $0 \le k \le N$ and $0.5<p_l<p_h\le 1$.
	%\footnote{The analysis will be tricky when there are more than two solution accuracy levels \cite{fernandez2010closed}. 
	%	We plan to study the general case in future work.}  
	We use $\mathcal{N}_h$ and $\mathcal{N}_l$ to denote the set of high-accuracy and low-accuracy workers, respectively.
	%\vspace{-1mm}
%	\begin{equation}\label{effortmodel}
%	\begin{aligned}
%	P(x_{i}^{\rm estimate}=x)=
%	\begin{cases}
%	0.5, \quad &\text{if} \; e_{i}=0 \; (\text{with zero cost}),\\
%	p_{i}, \quad &\text{if} \; e_{i}=1 \; (\text{with a cost} \; c \ge 0),
%	\end{cases}
%	%\vspace{-1.5mm}
%	\end{aligned}
%	\end{equation}
 \subsubsection{Worker Solution Reporting Strategy}
	Each worker also needs to decide whether to truthfully report his  solution to the platform. 
	%After completing the task, each worker  $i$ generates an estimated solution  $x_{i}^{\rm estimate} \in \mathcal{X}$, and he can report a value $x_{i}^{\rm report} \in \mathcal{X}$ to the platform that may or may not be the same as $x_{i}^{\rm estimate}$.
	 For a worker $i$ who does not exert effort, he can only apply the random reporting strategy denoted by $r_i=\rm{rd}$.\footnote{The $\rm{rd}$ strategy is adopted for ease of exposition. In fact, if a worker exerts no effort, his solution is equally likely to be correct or wrong.  Hence, one could equivalently view that the worker is either truthfully reporting its solution or untruthfully reporting it.} For those workers who exert effort, they can either truthfully or untruthfully report their solutions, where
	% (inversing the solutions in the binary case) due to self-interest 
	we use $r_i \in \{1,-1\}$ to denote the reporting strategy  with $r_i = 1$ indicating truthful reporting and $r_i = -1$ indicating untruthful reporting.
	%\[v_i=
	%\begin{cases}
	%1 & \rightarrow \text{truthful reporting}\\
	%-1 & \rightarrow \text{untruthful reporting}
	%\end{cases}
	%\]
	 More specifically, 
	\begin{equation}
	x_i^{\rm report}=
	\begin{cases}
	x_i^{\rm estimate}, \quad & \text{if} \; r_i=1,\\
	-x_i^{\rm estimate}, \quad & \text{if} \; r_i=-1,\\
	1 \;\text{or}\; -1 \; \text{with an equal probability}, \quad & \text{if} \; r_i={\rm rd}.
	\end{cases}
	\end{equation}
In fact,  workers can benefit from colluding to always report  $1$ (or $-1$) as the task solution, under the majority voting scheme.  However, such colluding strategies require much coordination among workers. This may not be possible in an online crowdsourcing platform where workers are temporally and spatially separated with limited communications.  In addition, workers engaged in such strategies could be detected by the platform and removed. Hence, we consider that workers report their solutions independently and restrict their solution reporting strategies to $\{{\rm rd}, 1,-1\}$  \cite{liu2016learning,chaoGC19,chaoGS19}. 
	
	For ease of exposition, we use $s_i\triangleq (e_i, r_i)$ to denote worker $i$'s effort exertion and reporting strategy with  $s_i \in \mathcal{S}_i\triangleq\{(0, \rm{rd}), (1,1), (1,-1)\}$.
	
	%For notational convenience, we use $s_{i,t}\triangleq(e_{i,t}, r_i)$ to denote each worker's effort exertion and solution reporting strategy, where $s_{i,t}$ belongs to the set $\mathcal{S}=\{(0, 1), (0,-1), (1,1), (1,-1)\}$. 
	
	%\subsubsection{Internal Reward for Truthful Reports} In practice, workers have ethical preference for honesty and for providing useful solutions to the platform. To model this, we consider  that if an effort-exerting worker $i$ truthfully reports, he obtains an internal reward $l \ge 0$.   Otherwise, he will obtain an internal penalty of $-l$ from untruthfully reporting. \footnote{We can generalize our model to the case where the amounts of internal rewards and penalties are different.} If a worker does not exert effort and  hence randomly reports, the internal reward will be $0$ because this report is equally likely to help or hurt the platform. To avoid the trivial case, we assume $l< c$, i.e., to fully compensate a worker's cost of effort exertion, a worker still needs a positive consistency reward from the platform besides his own internal reward from truthfully reporting. 

	\subsubsection{Consistency Reward for Majority Voting} After workers complete the task and report their solutions, the platform compares each worker $i$'s reported solution $x_i^{\rm report}$ with the majority solution from the remaining workers. If they are aligned,
	%(denoted by $x^{\rm majority}_{-i}$)
	 worker $i$ will receive a \textit{consistency reward} $R \ge 0$.\footnote{In case a tie is incurred in the majority solution, we assume without loss of generality that, worker $i$ still obtains the reward.} 
	%Note that the reward $R(\epsilon^h)$ is a function of the platform's decision $\epsilon^h$ in Stage I, which will be explained in Section ?.  
%	Specifically, the majority solution from worker $i$'s perspective is
%	\begin{equation}
%	x^{\rm majority}_{-i}=\\
%	\begin{cases}
%	1, \quad & \text{if} \quad \sum_{j \in \mathcal{N}, j \neq i}x_{j}^{\rm report}>0,\\
%	-1, \quad & \text{if} \quad \sum_{j \in \mathcal{N}, j \neq i}x_{j}^{\rm report}<0,\\
%	{\rm tie}, \quad & \text{if} \quad \sum_{j \in \mathcal{N}, j \neq i}x_{j}^{\rm report}=0.
%	\end{cases}
%	\end{equation}
%	There are three cases to discuss: 
%	
%	\begin{itemize}
%		\item  If $x_i^{\rm report} = x^{\rm majority}_{-i}$, worker $i$ receives $R$ for matching the majority solution. 
%		
%		\item If $x^{\rm majority}_{-i}={\rm tie}$, worker $i$ also receives $R$ because his reported solution determines the majority solution from all the workers. 
%		
%		\item  If $x_i^{\rm report} \neq x^{\rm majority}_{-i}$ and $x^{\rm majority}_{-i}\neq{\rm tie}$, worker $i$ will not receive any consistency reward $R$. 
%	\end{itemize}
	%Clearly worker $i$'s consistency reward depends on other workers' reporting strategies. 
	We use $G_{i}(\boldsymbol{s};\boldsymbol{\epsilon})$ to denote the probability of worker $i$ receiving $R$,  where 
	$\boldsymbol{s}=((e_{i}, r_i), \forall i \in \mathcal{N})$.  Note that $G_{i}\left(\boldsymbol{s};\boldsymbol{\epsilon}\right)$ is also a function of the platform's information revelation strategy $\boldsymbol{\epsilon}$, which will be explained in Section \ref{platform_problem}.

	\subsubsection{Worker Payoff} We define each worker $i$'s expected  payoff  as 
	\begin{equation}\label{worker_payoff}
	\begin{aligned}
	u_i\left( \boldsymbol{s}; \boldsymbol{\epsilon}, R\right)=  G_{i}\left(\boldsymbol{s}; \boldsymbol{\epsilon}\right) \cdot R-e_{i} \cdot c,
	\end{aligned}
	\end{equation}
	where $ G_{i}\left(\boldsymbol{s};\boldsymbol{\epsilon}\right) \cdot R$ represents the expected consistency reward 
	%$r_i\cdot l$ represents the internal reward,\footnote{According to our previous definition, when $r_i={\rm rd}$, the product between $r_i$ and $l$ is zero.}
	 and $-e_{i} \cdot c$ represents the cost for effort exertion.
	
	\subsection{Platform's Decisions and Payoff}\label{platform_problem}
	In this subsection, we first define the platform's strategies for information revelation and reward design, and then define its payoff function.
	
	\subsubsection{Platform Information Revelation Strategy} The platform has additional information regarding the workers' solution accuracy distribution, i.e., the number of high-accuracy workers $k$. With this informational advantage, the platform can strategically announce the worker
 information to induce desired worker behavior and maximize its payoff.
	
	The information asymmetry and revelation between the platform and the workers is modeled using a Bayesian persuasion framework \cite{crawford1982strategic,argenziano2016strategic,kamenica2011bayesian,kurschilgen2019communication} as follows:
	\begin{itemize}
		\item Step 1:  neither the platform nor the workers know $k$, and the platform must commit to a long-term information revelation strategy.
		\item Step 2: the value of $k$ is realized and observed by the platform, but not the workers.
		\item Step 3: the platform announces a value $k_p^{\rm anu}$ that may be different from the real value of $k$ to the workers according to its previously committed strategy.
	\end{itemize}    
	
	Next, we elaborate Step 1 to Step 3 in the following:
	
	\textbf{Step 1:} neither the platform nor workers know $k$. We consider they know the distribution of $k$, which constitutes their common prior belief $\boldsymbol{\mu}^{\rm prior}=\left(\mu^{\rm prior}_{\rm high}, \mu^{\rm prior}_{\rm low}\right)$, where
	\begin{equation}\label{prior}
	\mu^{\rm prior}_{\rm high}=Pr\left(k=k^{\rm high}\right), \;\; \mu^{\rm prior}_{\rm low}=Pr\left(k=k^{\rm low}\right),
	\end{equation}
	$\mu^{\rm prior}_{\rm high}+\mu^{\rm prior}_{\rm low}=1$, and  $k^{\rm high}$ and $k^{\rm low}$ are two possible values of $k$ satisfying $0\le k^{\rm low}< k^{\rm high}\le N$. Note that $\boldsymbol{\mu}^{\rm prior}$ represents the workers' prior belief \textit{before} the platform announces $k_p^{\rm anu}$. In practice, the workers can form such a prior belief via exploring the platform's feedback and reputation systems \cite{jagabathula2014reputation}.
	\textit{It is important to note that in (\ref{prior}), the consideration of a two-point distribution for $k$ is just for the ease of exposition. Our analysis and results are applicable to the case when there are arbitrarily finite realizations for $k$. }
	
	Before $k$ is realized, the platform commits to an information revelation strategy. The platform may have many tasks in practice, and for each task, it may face a different worker population (i.e., different $k$). Before workers arrive (i.e., before $k$ is realized), the platform determines the information revelation strategy, and will commit to this long-term strategy, which helps it build a good reputation \cite{kreps1982reputation,kamenica2011bayesian,camara2019avoiding,aumann1995repeated,alonso2008optimal}.\footnote{Information revelation in organizations often involve commitment, formally through contracts, or informally through reputation \cite{kamenica2011bayesian}. Nevertheless, we can also analyze the case where the platform can choose different information revelation strategies for different tasks and worker populations.} 
	
	Specifically, we use $\boldsymbol{\epsilon}\triangleq\left(\epsilon^h, \epsilon^l\right) \in [0,1]^2$ to denote the platform's information revelation strategy. We assume that the platform announces $k_{p}^{\rm anu}=k^{\rm high}$ with probability $\epsilon^h$ when $k=k^{\rm low}$.  
	%One can understand $\epsilon^h$ as the probability of deception when the platform announces $k_p^{\rm anu}=k^{\rm high}$. 
	Moreover, we assume that the platform announces $k_p^{\rm anu}=k^{\rm low}$ with probability $\epsilon^l$ when $k=k^{\rm high}$.\footnote{Considering a positive $\epsilon^l$ seems to be counter-intuitive, because one might expect that the platform should reveal a high overall worker capability (e.g., announce $k^{\rm high}$ instead of $k^{\rm low}$) to impress the workers. An analogous example is that a seller would  exaggerate the quality of its item to attract more buyers. We will later show that the platform indeed may have the incentive to choose a positive $\epsilon^l$ in some cases.}  Both $\epsilon^h$ and $\epsilon^l$ are the platform's deception probabilities in different cases. 
	%Nevertheless, the consideration of $\epsilon^l$ as a platform decision besides $\epsilon^h$ makes our model most general. 
  We summarize the dependence of $k_p^{\rm anu}$ on $k$ and $\boldsymbol{\epsilon}$  as follows:
	\begin{equation}\label{epsilonh}
	\begin{aligned}
	\begin{cases}
	&Pr\left(k_p^{\rm anu}=k^{\rm high}|k=k^{\rm low}\right)=\epsilon^h,\\     
	&Pr\left(k_p^{\rm anu}=k^{\rm low}|k=k^{\rm low}\right)=1-\epsilon^h,\\
	&Pr\left(k_p^{\rm anu}=k^{\rm high}|k=k^{\rm high}\right)=1-\epsilon^l,\\
	&Pr\left(k_p^{\rm anu}=k^{\rm low}|k=k^{\rm high}\right)=\epsilon^l.
	\end{cases}
	\end{aligned}
	\end{equation} 
	%As will be shown, the platform's information revelation problem is non-convex and challenging to solve. Nevertheless, we can exploit the special structure of the problem to characterize the optimal information revelation strategies.
%	
%	\textbf{In (\ref{epsilonh}), we consider that when $k=k^{\rm high}$, the platform will not announce $k_p^{\rm anu}=k^{\rm low}$. This is consistent with the common intuition that in practice, the platform usually does not announce a lower system performance than its actual value. Nevertheless, we will study the general case in Section V to investigate whether the platform has incentives to announce $k^{\rm low}$ when $k=k^{\rm high}$.}
	
	\textbf{Step 2:} $k$ is realized according to its distribution (which is also the common prior) $\boldsymbol{\mu}^{\rm prior}$ \cite{kamenica2011bayesian,frug2018strategic} and the value of $k$ is observed by the platform (e.g., via market research), but not the workers. More specifically, once the crowdsourcing task is announced and workers have signed up the task, the platform can know the workers' capabilities (i.e., the value of $k$) from their past performances.
	
	\textbf{Step 3}: given $k$, the platform will announce $k_p^{\rm anu}$ to the workers according to the committed strategy specified in (\ref{epsilonh}). The platform will strategically choose $\boldsymbol{\epsilon}$, which  affects the workers' posterior belief regarding $k$, and hence affects the platform's reward design and payoff. Note that  workers can learn the platform's information revelation strategy via repeated interactions with the platform \cite{tang2019incentive}. They can also learn such strategy via exploring the platform's feedback and reputation systems \cite{yu2020crowdr}.
	\subsubsection{Platform Reward Design Strategy}
	In addition to the information revelation strategy $\boldsymbol{\epsilon}$, the platform also needs to decide the consistency reward per worker $R\ge 0$ to incentivize high-quality and truthful solutions. As mentioned earlier, after workers report their solutions, the platform will distribute consistency rewards to workers whose solutions match the majority. The decisions $\boldsymbol{\epsilon}$ and $R$ are complexly coupled, as $\boldsymbol{\epsilon}$ affects the workers' posterior belief, which together with $R$ determine the worker equilibrium outcome.
	
	\subsubsection{Platform Payoff}
	%The platform will first choose $\epsilon^h \in [0,1]$ and then $R\ge 0$ to maximize its expected payoff, anticipating the worker equilibrium behaviours. 
	The platform aims to achieve a good tradeoff between the accuracy (i.e., quality) of the aggregated solution and the cost of incentivizing the workers \cite{liu2016learning,chaoGC19,chaoGS19,9244604}. Specifically, we define the platform's payoff as follows:

	\vspace{-2mm}
	\begin{equation}\label{platformpayofffunction}
	U_p (\boldsymbol{\epsilon}, R, k; \boldsymbol{s})= \beta  P_a(\boldsymbol{\epsilon},R,k; \boldsymbol{s}) - {\mathbb E} \left\{ R^{tot} (\boldsymbol{\epsilon}, R,k; \boldsymbol{s})\right\},
	\end{equation} 
	where $P_a(\boldsymbol{\epsilon},R,k; \boldsymbol{s}) $ denotes the accuracy of the aggregated solution from the workers, i.e., the probability that the aggregated solution matches the true task solution. We will use the widely adopted \textit{majority rule} \cite{8387487,chaoGC19} to calculate this probability. The parameter $\beta >0$ represents the platform's valuation of  the aggregated solution's accuracy. The term ${\mathbb E} \left\{ R^{tot} (\boldsymbol{\epsilon}, R,k; \boldsymbol{s})\right\}$ captures the total expected consistency rewards.

	\vspace{-0.2mm}
	
	\section{Solving Three-Stage Model}\label{naiveworkers}
	In this section, we solve the three-stage game via backward induction. 
	% considering naive workers, i.e., the workers will trust any $k_p^{\rm anu}$ announced by the platform and use it as the posterior belief. The consideration of naive workers serves as a benchmark and also has important practical implications. For example, it can model the scenario where workers are confident in the platform's released information (especially for those platforms with good reputations). It can also model the scenario where the workers do not have enough rationality to reason about the actual $k$ based on the announcement  \cite{shao2019multimedia}. 
	Specifically, we will solve the workers' decisions, the platform's reward design, and the platform's information revelation in Sections \ref{stage3}, \ref{stage2}, and \ref{stage1}, respectively. 
	
	\textit{To save space, we only consider strategic workers in this section, as the analysis is more challenging and results are more intriguing. We will later provide numerical results for naive workers in Section \ref{numerical}, which serves as a benchmark.}
	
	%Since the focus of this paper is to explore the platform's information revelation, to save space,  we mainly discuss strategic workers when presenting solutions to Stages III and II. When we solve the platform's information revelation in Stage I, we provide details for both strategic workers and naive workers.
	\subsection{Worker Equilibrium Behaviors in Stage III}\label{stage3}
	Given $\boldsymbol{\epsilon}$ and $R$, each worker chooses his effort exertion and solution reporting strategies $s_i$ to maximize his own payoff.
	\subsubsection{Worker Belief Update}  
	Before the platform announces $k_p^{\rm anu}$, workers have prior belief $\boldsymbol{\mu}^{\rm prior}$ regarding $k$ (see (\ref{prior})). After the platform announces $k_p^{\rm anu}$, the workers form posterior belief based on the prior belief and the announcement.  Proposition \ref{stg_post} computes the posterior belief for strategic workers.
	
	\begin{Pro} \label{stg_post}{(Posterior Belief for Strategic Workers)} 
		Let $\mu^{\rm post, str}_{w}|_{k_p^{\rm anu}}$ denote a strategic worker's  posterior belief in $k^w$ conditional on the platform's announcement $k_p^{\rm anu}$, where $w \in \left\{{\rm high}, {\rm low}\right\}$. Then, we have
%		Let  $\boldsymbol{\mu}^{\rm post, str}\left(\boldsymbol{\epsilon}\right)\triangleq \left(\boldsymbol{\mu}^{\rm post, str}_{\rm high}\left(\boldsymbol{\epsilon}\right), \boldsymbol{\mu}^{\rm post, str}_{\rm low}\left(\boldsymbol{\epsilon}\right)\right)$  denote a strategic worker's posterior belief, where 
%		\begin{equation} 
%		\boldsymbol{\mu}^{\rm post, str}_{\rm high}\left(\boldsymbol{\epsilon}\right)=\left(\mu^{\rm post, str}_{\rm high}|_{k^{\rm high}}\left(\boldsymbol{\epsilon}\right),  \mu^{\rm post, str}_{\rm high}|_{k^{\rm low}}\left(\boldsymbol{\epsilon}\right)\right),
%		\end{equation}
%		\begin{equation}
%		\boldsymbol{\mu}^{\rm post, str}_{\rm low}\left(\boldsymbol{\epsilon}\right)=\left(\mu^{\rm post, str}_{\rm low}|_{k^{\rm high}}\left(\boldsymbol{\epsilon}\right),  \mu^{\rm post, str}_{\rm low}|_{k^{\rm low}}\left(\boldsymbol{\epsilon}\right)\right).
%	\end{equation}
	%Based on Bayes' rule, We compute the conditional probabilities that represent a strategic worker's posterior belief as follows:
	\begin{equation}\label{belief_h_h}
	\begin{aligned}
	\mu^{\rm post, str}_{\rm high}|_{k^{\rm high}}\left(\boldsymbol{\epsilon}\right)
	%=Pr(k=k^{\rm high}|k_p^{\rm anu}=k^{
	%	\rm high})\\
	%&=\frac{Pr(k_p^{\rm anu}=k^{\rm high}|k=k^{\rm high})\cdot Pr(k=k^{\rm high})}{Pr(k_p^{\rm anu}=k^{\rm high})}\\
	%&=\frac{Pr(k_p^{\rm anu}=k^{\rm high}|k=k^{\rm high})\cdot Pr(k=k^{\rm high})}{\sum_{\bar{k}\in \left\{k^{\rm high}, k^{\rm low}\right\}} \left(Pr(k_p^{\rm anu}=k^{\rm high}|k=\bar{k})\cdot Pr(k=\bar{k})\right)}\\
	=\frac{\left(1-\epsilon^l\right)\mu^{\rm prior}_{\rm high}}{ \left(1-\epsilon^l\right)\mu^{\rm prior}_{\rm high}+\epsilon^h\mu^{\rm prior}_{\rm low}},
	\end{aligned}
	\end{equation}
		\begin{equation}\label{post_l_h}
	\begin{aligned}
	\mu^{\rm post, str}_{\rm low}|_{k^{\rm high}}\left(\boldsymbol{\epsilon}\right)
	%=Pr(k=k^{\rm low}|k_p^{\rm anu}=k^{
	%	\rm high})\\
	%&=1-\mu^{\rm post, str}_{\rm high}|_{k^{\rm high}}\left(\boldsymbol{\epsilon}\right)
	=\frac{\epsilon^h\mu^{\rm prior}_{\rm low}}{ \left(1-\epsilon^l\right)\mu^{\rm prior}_{\rm high}+\epsilon^h\mu^{\rm prior}_{\rm low}},
	\end{aligned}
	\end{equation}
		\begin{equation}\label{post_l,l}
	\begin{aligned}
	\mu^{\rm post, str}_{\rm high}|_{k^{\rm low}}\left(\boldsymbol{\epsilon}\right)
	%=Pr(k=k^{\rm high}|k_p^{\rm anu}=k^{
	%	\rm low})\\
	%&=\frac{Pr(k_p^{\rm anu}=k^{\rm low}|k=k^{\rm high})\cdot Pr(k=k^{\rm high})}{Pr(k_p^{\rm anu}=k^{\rm low})}\\
	%=\frac{Pr(k_p^{\rm anu}=k^{\rm low}|k=k^{\rm high})\cdot Pr(k=k^{\rm high})}{\sum_{\bar{k}\in \left\{k^{\rm high}, k^{\rm low}\right\}} \left(Pr(k_p^{\rm anu}=k^{\rm low}|k=\bar{k})\cdot Pr(k=\bar{k})\right)}\\
	=\frac{\epsilon^l\mu^{\rm prior}_{\rm high}}{ \epsilon^l\mu^{\rm prior}_{\rm high}+\left(1-\epsilon^h\right)\mu^{\rm prior}_{\rm low}},
	\end{aligned}
	\end{equation}
	\begin{equation}\label{belief_l_l}
	\begin{aligned}
	\mu^{\rm post, str}_{\rm low}|_{k^{\rm low}}\left(\boldsymbol{\epsilon}\right)
	%=Pr(k=k^{\rm low}|k_p^{\rm anu}=k^{
	%	\rm low})\\
	%&=1-\mu^{\rm post, str}_{\rm high}|_{k^{\rm low}}\left(\boldsymbol{\epsilon}\right)
	=\frac{\left(1-\epsilon^h\right)\mu^{\rm prior}_{\rm low}}{ \epsilon^l\mu^{\rm prior}_{\rm high}+\left(1-\epsilon^h\right)\mu^{\rm prior}_{\rm low}}.
	\end{aligned}
	\end{equation}
	\end{Pro}
\begin{proof}
	We apply the Bayes' rule to compute the above probabilities.  We show the derivation of (\ref{belief_h_h}) as follows:
	\begin{displaymath}
	\begin{aligned}
    &\mu^{\rm post, str}_{\rm high}|_{k^{\rm high}}\left(\boldsymbol{\epsilon}\right)=Pr(k=k^{\rm high}|k_p^{\rm anu}=k^{\rm high})\\
    &=\frac{Pr(k_p^{\rm anu}=k^{\rm high}|k=k^{\rm high})\cdot Pr(k=k^{\rm high})}{Pr(k_p^{\rm anu}=k^{\rm high})}\\
    &=\frac{Pr(k_p^{\rm anu}=k^{\rm high}|k=k^{\rm high})\cdot Pr(k=k^{\rm high})}{\sum_{\bar{k}\in \left\{k^{\rm high}, k^{\rm low}\right\}} \left(Pr(k_p^{\rm anu}=k^{\rm high}|k=\bar{k})\cdot Pr(k=\bar{k})\right)}\\
    &=\frac{\left(1-\epsilon^l\right)\mu^{\rm prior}_{\rm high}}{ \left(1-\epsilon^l\right)\mu^{\rm prior}_{\rm high}+\epsilon^h\mu^{\rm prior}_{\rm low}}.
	\end{aligned}
	\end{displaymath}
	Similarly, one can compute the probabilities in (\ref{post_l_h})-(\ref{belief_l_l}). 
\end{proof}
	 Note that $\mu^{\rm post, str}_{\rm high}|_{k^{\rm high}}\left(\boldsymbol{\epsilon}\right)$ in (\ref{belief_h_h}) represents the posterior belief of a strategic worker regarding $k=k^{\rm high}$ conditional on the platform's announcement $k_{p}^{\rm anu}=k^{\rm high}$, and it is a function of $\boldsymbol{\epsilon}$  and the prior $\boldsymbol{\mu}^{\rm prior}$.  One can see that  $\mu^{\rm post, str}_{\rm high}|_{k^{\rm high}}\left(\boldsymbol{\epsilon}\right)$ in (\ref{belief_h_h}) decreases in $\epsilon^h$. This indicates that if the platform is more likely to deceive the workers via announcing $k_{p}^{\rm anu}=k^{\rm high}$ when  $k=k^{\rm low}$ (i.e., a larger $\epsilon^h$), the strategic workers hearing $k_p^{\rm anu}=k^{\rm high}$ will doubt the platform and are less likely to believe  $k=k^{\rm high}$ (which leads to a smaller $\mu^{\rm post, str}_{\rm high}|_{k^{\rm high}}\left(\boldsymbol{\epsilon}\right)$).  Similarly, $\mu^{\rm post, str}_{\rm low}|_{k^{\rm low}}\left(\boldsymbol{\epsilon}\right)$  in (\ref{belief_l_l}) decreases in $\epsilon^l$. If the platform is more likely to lie announcing $k_p^{\rm anu}=k^{\rm low}$, the strategic workers hearing this are less likely to believe $k=k^{\rm low}$ (i.e., a smaller $\mu^{\rm post, str}_{\rm low}|_{k^{\rm low}}\left(\boldsymbol{\epsilon}\right)$).

	\subsubsection{Worker Equilibrium Strategy}
	When the workers play the effort exertion and solution reporting game, they use their posterior belief to calculate their expected payoffs.
	%The equilibrium analyses for the above four cases have similar structures, and we will only use \textbf{Case h,h} as an illustration and save space.
		Next, we characterize the workers' equilibrium decisions. Similar to prior IEWV literature (e.g., \cite{liu2016learning,dasgupta2013crowdsourced,chaoGC19,chaoWiOpt20,chaoGS19}), we focus on symmetric Nash equilibria (SNE), where workers with the same type (i.e., solution accuracy) play the same strategy. 
	
	\begin{figure*}
		\vspace{-2mm}
		\begin{equation}\label{condition_psne}
		\begin{aligned}
		\frac{2p_h-1}{2p_l-1}\left(\mu^{\rm post, str}_{\rm high}|_{k^{\rm anu}_p}\left(\boldsymbol{\epsilon}\right)P^{\rm majority}_{k^{\rm high}-1}+\mu^{\rm post, str}_{\rm low}|_{k^{\rm anu}_p}\left(\boldsymbol{\epsilon}\right)P^{\rm majority}_{k^{\rm low}-1}\right)
		\ge \mu^{\rm post, str}_{\rm high}|_{k^{\rm anu}_p}\left(\boldsymbol{\epsilon}\right)P^{\rm majority}_{\rm high}+\mu^{\rm post, str}_{\rm low}|_{k^{\rm anu}_p}\left(\boldsymbol{\epsilon}\right) P^{\rm majority}_{k^{\rm low}}.
		\end{aligned}
		\end{equation}
		\rule[-12pt]{\textwidth}{0.05em}
		\vspace{-5mm}
	\end{figure*}
	
	For the ease of exposition, we first define some terminology related to the worker equilibria in Stage III.
	\begin{Def}\label{def}{(Stage III Equilibrium Types)}
		%\begin{itemize}
		%	\vspace{-1mm}
		(i) An $n$-SNE is defined as the  profile $(s_i^*=(0, {\rm{rd}}), \forall i\in \mathcal{N})$, where no worker exerts effort and truthfully reports.\\
		(ii) An  $f$-SNE is defined as the profile $(s_i^*=(1,1), \forall i\in \mathcal{N})$, where all the workers {\it fully} exert effort and truthfully report.\\
		(iii) A $p$-SNE is defined as the profile $(s_i^*=(1,1), \forall i \in \mathcal{N}_h, s_j^*=(0,{\rm{rd}}), \forall j \in \mathcal{N}_l)$, where high-accuracy workers exert effort and truthfully report, and low-accuracy workers exert no effort and randomly report.
		%\footnote{Through Definition \ref{def} and henceforth, SNE means symmetric Nash equilibrium.}
		%\end{itemize}
	\end{Def}
	
%	Depending on the true value of $k$ and the platform's announcement, there are four cases on the worker equilibrium analysis,\footnote{Note that the true value of $k$ affects the platform's reward design in Stage II, and hence affects the worker equilibrium outcome in Stage III.} which are defined as follows.
%	\begin{itemize}
%		\item \textbf{Case h,h}: $k=k^{\rm high}$ and $k_{p}^{\rm anu}=k^{\rm high}$,  which occurs with probability $\mu_{\rm prior}^{\rm high}$.
%		\item \textbf{Case h,l}: $k=k^{\rm high}$ but $k_{p}^{\rm anu}=k^{\rm low}$, which occurs with probability $0$.
%		\item \textbf{Case l,h}: $k=k^{\rm low}$ but  $k_{p}^{\rm anu}=k^{\rm high}$, which occurs with probability $\mu_{\rm prior}^{\rm low} \cdot \epsilon^h$.
%		\item \textbf{Case l,l}: $k=k^{\rm low}$ and $k_{p}^{\rm anu}=k^{\rm low}$, which occurs with probability $\mu_{\rm prior}^{\rm low} \cdot \left(1-\epsilon^h\right)$.
%	\end{itemize} 
% 

%Note that $\boldsymbol{\epsilon}$ affects the strategic workers' posterior beliefs but not that of the naive workers. Hence, the worker equilibrium analysis is more challenging in the strategic worker case. \textit{To save space, we only characterize the equilibria among strategics workers, and the equilibira among naive workers are omitted.}

Note that the workers’ equilibrium behaviors depend on the platform’s
information announcement $k_p^{\rm anu}$, which determines the workers’
posterior belief. Theorem \ref{Equilibrium_SNE} characterizes the possible equilibria among strategic workers.
% when $k_p^{\rm anu}=k^{\rm high}$. \textit{The analysis and results when $k_p^{\rm anu}=k^{\rm low}$ are similar and also omitted due to space limitations.}

	\begin{The}\label{Equilibrium_SNE}{(Worker Equilibria in Stage III)} \\
	 (i) Given any $\boldsymbol{\epsilon}\in [0,1]^2$, an $n$-SNE exists if $R \ge 0$.\\
	 (ii) Given any $\boldsymbol{\epsilon}\in [0,1]^2$, there always exists a threshold $R_f^{\rm str}\left(\boldsymbol{\epsilon}, k_p^{\rm anu}\right)>0$, such that an $f$-SNE exists if and only if $R \ge R_f^{\rm str}\left(\boldsymbol{\epsilon}, k_p^{\rm anu}\right)$.\\
	 (iii) When $\boldsymbol{\epsilon} \in \Phi \triangleq \left\{\boldsymbol{\epsilon} \in [0,1]^2 \big|\; \text{condition}\; (\ref{condition_psne})  \; \text{holds}\right\}$, there exist two thresholds $0<R_{pl}^{\rm str}\left(\boldsymbol{\epsilon}, k^{\rm anu}_p\right)\leq R_{ph}^{\rm str}\left(\boldsymbol{\epsilon},k^{\rm anu}_p\right)$, such that a $p$-SNE exists if and only if $R_{pl}^{\rm str}\left(\boldsymbol{\epsilon},k^{\rm anu}_p\right)\le R \le R_{ph}^{\rm str}\left(\boldsymbol{\epsilon},k^{\rm anu}_p\right)$, where the condition (\ref{condition_psne}) is shown at the top of the next page.	
	\end{The}

	 Cases $(i)$ and $(ii)$ in Theorem \ref{Equilibrium_SNE} implies that given any  $\boldsymbol{\epsilon}$, an $n$-SNE and an $f$-SNE can exist under proper reward levels. In particular, they can coexist under a sufficiently large reward level (i.e., $R\ge R_f^{\rm str}\left(\boldsymbol{\epsilon}, k_p^{\rm anu}\right)$).
	 
	 Case $(iii)$ in Theorem \ref{Equilibrium_SNE} suggests a trickier analysis of the existence of $p$-SNE. Unlike $n$-SNE and $f$-SNE, for a $p$-SNE to exist, the information revelation strategy $\boldsymbol{\epsilon}$ must satisfy the condition specified in (\ref{condition_psne}). In (\ref{condition_psne}), $P^{\rm majority}_{k^{\rm high}-1}$ is the probability that the majority solution among $N-1$ workers is correct when $k^{\rm high}-1$ high-accuracy workers use $(1,1)$ and the remaining  workers use $(0, {\rm rd})$.  Similar discussions apply for $P^{\rm majority}_{k^{\rm low}-1}$, $P^{\rm majority}_{\rm high}$, and $P^{\rm majority}_{k^{\rm low}}$ in (\ref{condition_psne}). The expressions of these probability terms are complicated, yet we can apply the Z-transform algorithm in \cite{fernandez2010closed} to efficiently calculate them. Condition (\ref{condition_psne}) means that when exerting efforts, the high-accuracy workers believe they have a larger probability of obtaining the reward than the low-accuracy workers. If (\ref{condition_psne}) is violated, the high-accuracy workers will believe the chance of obtaining the reward is small. Together with a moderate amount of reward (i.e., $R_{pl}^{\rm str}\left(\boldsymbol{\epsilon}, k_p^{\rm anu}\right)\le R \le R_{ph}^{\rm str}\left(\boldsymbol{\epsilon}, k_p^{\rm anu}\right)$), the expected reward for a high-accuracy worker will be small. In this case,  the high-accuracy workers will not exert effort to save the cost, and  hence a $p$-SNE does not exist. 
	
Next, we characterize the impact of the platform's information revelation $\boldsymbol{\epsilon}$ on the worker equilibria in Corollary \ref{Reward_epsilonh}.
	
	\begin{Cor}\label{Reward_epsilonh}{}
		(i) Consider a fixed $\epsilon^l$. If $k_p^{\rm anu}=k^{\rm high}$,   $R_f^{\rm str}\left(\boldsymbol{\epsilon}, k_p^{\rm anu}\right)$ and $R_{pl}^{\rm str}\left(\boldsymbol{\epsilon}, k_p^{\rm anu}\right)$ in Theorem \ref{Equilibrium_SNE}  increase in $\epsilon^h$. Otherwise, if $k_p^{\rm anu}=k^{\rm low}$, both values decrease in $\epsilon^h$.\\
		(ii) Consider a fixed $\epsilon^h$. If $k_p^{\rm anu}=k^{\rm high}$,  $R_f^{\rm str}\left(\boldsymbol{\epsilon}, k_p^{\rm anu}\right)$ and $R_{pl}^{\rm str}\left(\boldsymbol{\epsilon}, k_p^{\rm anu}\right)$ in Theorem \ref{Equilibrium_SNE} also increase in $\epsilon^l$. Otherwise, if $k_p^{\rm anu}=k^{\rm low}$, both values decrease in $\epsilon^l$.
		
		\end{Cor}
	
	Case (i) in Corollary \ref{Reward_epsilonh} implies that if the platform is more likely to misreport $k_p^{\rm anu}=k^{\rm high}$ (i.e., a larger $\epsilon^h$), and it indeed announces $k^{\rm high}$, larger rewards are needed to induce both an $f$-SNE and a $p$-SNE. The intuition is that a larger $\epsilon^h$ makes the strategic workers believe that the average worker accuracy is lower (a smaller $\mu^{\rm post, str}_{\rm high}|_{k^{\rm high}}\left(\boldsymbol{\epsilon}\right)$ in (\ref{belief_h_h})), and hence there is a smaller chance of matching the majority solution and obtaining the reward. As a result, the platform needs to use larger rewards to incentivize the workers. Interestingly, Case (ii) says that $R_f^{\rm str}\left(\boldsymbol{\epsilon}, k_p^{\rm anu}\right)$ and $R_{pl}^{\rm str}\left(\boldsymbol{\epsilon}, k_p^{\rm anu}\right)$ also increase in $\epsilon^l$ when $k^{\rm anu}_p=k^{\rm high}$. A larger $\epsilon^l$ indicates that the platform is less likely to announce $k_{p}^{\rm anu}=k^{\rm high}$ when $k=k^{\rm high}$. Hence, when hearing $k_p^{\rm anu}=k^{\rm high}$, the strategic workers deduce that the real value is more likely to be $k=k^{\rm low}$. This implies a lower average worker accuracy and hence the platform needs larger rewards to incentivize them. 

	 Note that the different SNEs in Theorem \ref{Equilibrium_SNE} can coexist under proper conditions on $\boldsymbol{\epsilon}$ and $R$. For example, given any $\boldsymbol{\epsilon}$ and a sufficiently large $R$, at least an $n$-SNE and an $f$-SNE coexist. When multiple SNEs are possible, we are interested in understanding whether there exists a Pareto-dominant SNE, where each worker achieves a no smaller payoff,  with at least one worker achieving a strictly larger payoff, compared to that achieved in other possible SNEs \cite{kandori1993learning, chaoGC19,argenziano2016strategic}. We prove the existence of Pareto-dominant SNE in Theorem \ref{pareto}.
	\begin{The}\label{pareto}{(Pareto-dominant SNE in Stage III)}
	 Given any  $\boldsymbol{\epsilon}$ and $R$, there exists a Pareto-dominant SNE among the workers.
	\end{The}

 We assume that when multiple SNEs coexist, the workers will choose the Pareto-dominant one \cite{chaoGC19,kong2016equilibrium,harsanyi1988general}. Hence,  Theorem \ref{pareto} enables us to solve the platform's reward design and information revelation by focusing on the Pareto-dominant SNE among the workers.
	%\footnote{Besides the equilibria mentioned in Proposition \ref{Equilibrium_SNE}, there may exist other equilibria, i.e., $(s_i=(1,-1), \forall i)$, and $(s_i^*=(1,-1), \forall i \in \mathcal{N}_h, s_j^*=(0,{\rm{rd}}), \forall j \in \mathcal{N}_l)$. Nevertheless, the former is Pareto dominated by $f$-SNE and the latter by $p$-SNE. }
	
	\subsection{Platform Reward Design in Stage II}\label{stage2}
	In this subsection, we solve the platform's reward design problem. Given the decided $\boldsymbol{\epsilon}$ in Stage I, the platform observes $k$ and  decides $R$ in Stage II
	% to elicit $\boldsymbol{s^*} \in \{(s_i=(0, {\rm{rd}}), \forall i), (s_i=(1,1), \forall i), (s_i=(0,{\rm{rd}}), i \in \mathcal{N}_l, s_j=(1,1), j \in \mathcal{N}_h)\}$ 
	to maximize its payoff, anticipating the  Pareto-dominant SNE in Stage III. 
	%Again, we focus on the scenario when the platform announces in the strategic worker case. The analysis and results in other cases and those in the naive worker case are similar and omitted due to space limitation.
	
%	Since the consideration of naive workers serves as a benchmark and the focus of this paper is to explore the platform's information revelation strategy $\epsilon^h$, we omit the details on the platform's reward design in Stage II. 
%	%We will provide more details on the reward design in Section IV when we consider strategic workers. 
%	Nevertheless, the idea of the reward design is to compute the optimal reward to induce the desired Pareto-dominant equilibrium that maximizes the platform's payoff.\footnote{The intuition on the reward design for the four cases, i.e., \textbf{Case h,h}, \textbf{Case h,l}, \textbf{Case l,h}, and \textbf{Case l,l} is similar and also omitted.} 

	Before characterizing the platform's optimal reward design, we provide several definitions for ease of exposition.
	\begin{Def}\label{def_efficiency}
		Let $z \in \left\{n, f, p\right\}$ denote the SNE index. Define:
	\begin{itemize}[leftmargin=5mm]
		\item $P_z(k)$: accuracy of the aggregated solution for $z$-SNE.
		\item $\mathbb{E}\left\{R^{ tot}_z\left(\boldsymbol{\epsilon},k, k_p^{\rm anu}\right)\right\}$: total expected consistency reward for $z$-SNE.
		\item $B_z\left(\boldsymbol{\epsilon},k, k_p^{\rm anu}\right)$: Bang-per-buck for $z$-SNE, where
  	\begin{equation}
		\begin{aligned}
		B_z\left(\boldsymbol{\epsilon},k, k_p^{\rm anu}\right)=\frac{P_z\left(k\right)-P_n\left(k\right)}{\mathbb{E}\left\{R^{ tot}_z\left(\boldsymbol{\epsilon}, k, k_p^{\rm anu}\right)\right\}} .
		%&B_p\left(\boldsymbol{\epsilon},k, k_p^{\rm anu}\right)= (P_p\left(k\right)-P_n\left(k\right))/\mathbb{E}\left\{R^{ tot}_p\left(\boldsymbol{\epsilon},k, k_p^{\rm anu}\right)\right\},\\
		%&B_n= (P_n\left(k\right)-P_n\left(k\right))/\mathbb{E}\left\{R^{ tot}_n\left(\boldsymbol{\epsilon},k, k_p^{\rm anu}\right)\right\}=0.
		\end{aligned}
		\end{equation}
	\end{itemize} 
	\end{Def}

	One can think of $B_z\left(\boldsymbol{\epsilon},k, k_p^{\rm anu}\right)$ as the average accuracy improvement from $n$-SNE to $z$-SNE per unit of reward. 
	% and we name it the bang-per-buck for the $z$-SNE. where $z \in \left\{n, f, p\right\}$. 
	%It is trivial to show that $B_n=0$ due to zero accuracy improvement when no worker exerts effort. 
	As will be seen, the values of bang-per-buck for different SNEs will affect the platform's optimal reward design. 
%	One can consider $\mathbb{E}\left\{R^{ tot}_f\left(\boldsymbol{\epsilon},k, k_p^{\rm anu}\right)\right\}/(P_f\left(k\right)-P_n\left(k\right))$ (i.e., $1/B_f\left(\boldsymbol{\epsilon},k, k_p^{\rm anu}\right)$) as the average reward for per unit of accuracy improvement from $n$-SNE to $f$-SNE.  We use the reciprocal (i.e., $B_f\left(\boldsymbol{\epsilon},k, k_p^{\rm anu}\right)$) to reflect the equilibrium efficiency. The intuition is that the larger average reward is required for each unit of accuracy improvement, the less efficient is the equilibrium. Similar discussion applies for $B_p\left(\boldsymbol{\epsilon},k, k_p^{\rm anu}\right)$. It is trivial that $B_n=0$ due to no accuracy improvement when no worker exerts effort.

		Note that the terms in Definition \ref{def_efficiency} are also functions of $k$. In Stage II, the platform observes the real value of $k$, and will utilize it to optimize the reward design. Next, we characterize the platform's optimal reward design in Theorem \ref{PlatformDecision2}. 
	%\textit{To be consistent with the solutions shown in Stage III, here we only present the platform's optimal reward design for strategic workers when $k_p^{\rm anu}=k^{\rm high}$.}
		\begin{The}{(Platform's Reward Design in Stage II)}\label{PlatformDecision2}\\
		(i) If (\ref{condition_psne}) holds and $B_p\left(\boldsymbol{\epsilon},k, k_p^{\rm anu}\right) \ge B_f\left(\boldsymbol{\epsilon},k, k_p^{\rm anu}\right) $, the platform's optimal reward level is
		\begin{equation}\label{Optimal_R_case1}
		\begin{aligned}
		&R^{*}\left(\boldsymbol{\epsilon},k, k_p^{\rm anu}\right)=\\
		&\begin{cases}
		0, \quad &\text{{\rm{if}}} \quad \beta < \frac{1}{B_p\left(\boldsymbol{\epsilon}, k, k_p^{\rm anu}\right)},\\
		R_{pl}^{\rm str}\left(\boldsymbol{\epsilon}, k_p^{\rm anu}\right), \quad &\text{{\rm{if}}} \quad \frac{1}{B_p\left(\boldsymbol{\epsilon}, k, k_p^{\rm anu}\right)} \le \beta < \tilde{\beta}\left(\boldsymbol{\epsilon}, k, k_p^{\rm anu}\right) ,\\
		R_f^{\rm str}\left(\boldsymbol{\epsilon},k_p^{\rm anu}\right), \quad &\text{{\rm{if}}} \quad  \beta \ge \tilde{\beta}\left(\boldsymbol{\epsilon}, k, k_p^{\rm anu}\right),
		\end{cases}
		\end{aligned}
		\end{equation}
		where 
		$\tilde{\beta}\left(\boldsymbol{\epsilon}, k, k_p^{\rm anu}\right)= \frac{\mathbb{E}\left\{R^{tot}_f\left(\boldsymbol{\epsilon}, k, k_p^{\rm anu}\right)\right\}-\mathbb{E}\left\{R^{tot}_p\left(\boldsymbol{\epsilon}, k, k_p^{\rm anu}\right)\right\}}{P_f\left( k\right)-P_p\left( k\right)}$.\\
		(ii) If either (\ref{condition_psne}) does not hold
		 or $B_p\left(\boldsymbol{\epsilon}, k, k_p^{\rm anu}\right) < B_f\left(\boldsymbol{\epsilon}, k,k_p^{\rm anu}\right) $, the platform's optimal reward level is
		\begin{equation}\label{Optimal_R_case2}
		R^{*}\left(\boldsymbol{\epsilon}, k,k_p^{\rm anu} \right)=
		\begin{cases}
		0, \quad &\text{{\rm{if}}} \quad \beta < \frac{1}{B_f\left(\boldsymbol{\epsilon}, k,k_p^{\rm anu}\right)},\\
		R_f^{\rm str}\left(\boldsymbol{\epsilon},k_p^{\rm anu}\right), \quad &\text{{\rm{if}}} \quad  \beta \ge \frac{1}{B_f\left(\boldsymbol{\epsilon}, k,k_p^{\rm anu}\right)}.
		\end{cases}
		\end{equation}
		\end{The}
	
	Note that (\ref{condition_psne}) is a necessary condition for a $p$-SNE to exist. A $p$-SNE exists if (\ref{condition_psne}) holds and the reward level is appropriate (see Theorem \ref{Equilibrium_SNE}). Hence, Theorem \ref{PlatformDecision2} implies that: (i)  If a $p$-SNE exists, it has a larger bang-per-buck than an $f$-SNE, and the platform's valuation $\beta$ is moderate, the platform will elicit a $p$-SNE as the Pareto-dominant SNE (in Stage III) via choosing $R^*=R_{pl}^{\rm str}\left(\boldsymbol{\epsilon}, k_p^{\rm anu}\right)$ to maximize its payoff. (ii) Otherwise, if either a $p$-SNE does not exist or it has a smaller bang-per-buck than an $f$-SNE, a $p$-SNE cannot be optimal for the platform. When  $\beta$ is large, the platform will elicit an $f$-SNE as the Pareto-dominant SNE via choosing $R^*=R_f^{\rm str}\left(\boldsymbol{\epsilon},k_p^{\rm anu}\right)$. 

	\subsection{Platform Information Revelation in Stage I}\label{stage1}
	In this subsection, we solve the platform's information revelation problem. The platform decides the information revelation strategy $\boldsymbol{\epsilon}=\left(\epsilon^h, \epsilon^l\right) \in [0,1]^2$ in Stage I to maximize its expected payoff, anticipating its own decision on $R$ in Stage II and the workers' Pareto-dominant SNE in Stage III. Note that the expectation is taken with respect to $k$, as the platform needs to jointly consider different cases for $k$ to decide $\boldsymbol{\epsilon}$.
	
	%Before we present the Stage I solution, we introduce the information updates of the naive workers. The information updates for the strategic workers are given in (\ref{belief_h_h})-(\ref{belief_l_l}). Lemma \ref{naive_post} characterizes the naive workers' posterior beliefs. 

	The information revelation problem in Stage I is challenging to solve, as it is complexly coupled with the reward design in Stage II (e.g., as shown in (\ref{condition_psne})), which leads to a non-convex program. More specifically, the information revelation in Stage I affects the workers' posterior belief, which together with the reward design in Stage II determines the workers' equilibrium behaviors in Stage III. In addition, the workers' decisions in Stage III are intertwined, since each worker obtains a reward if his solution matches the majority solution of the other workers. Nevertheless, we can exploit the special property of the problem to explore the solutions to Stage I.
	
		\begin{figure*}[t]
			\vspace{-5mm}
	\centering
	\subfloat[Platform payoff vs. $p_h$. ]{\includegraphics[width=2.36in]{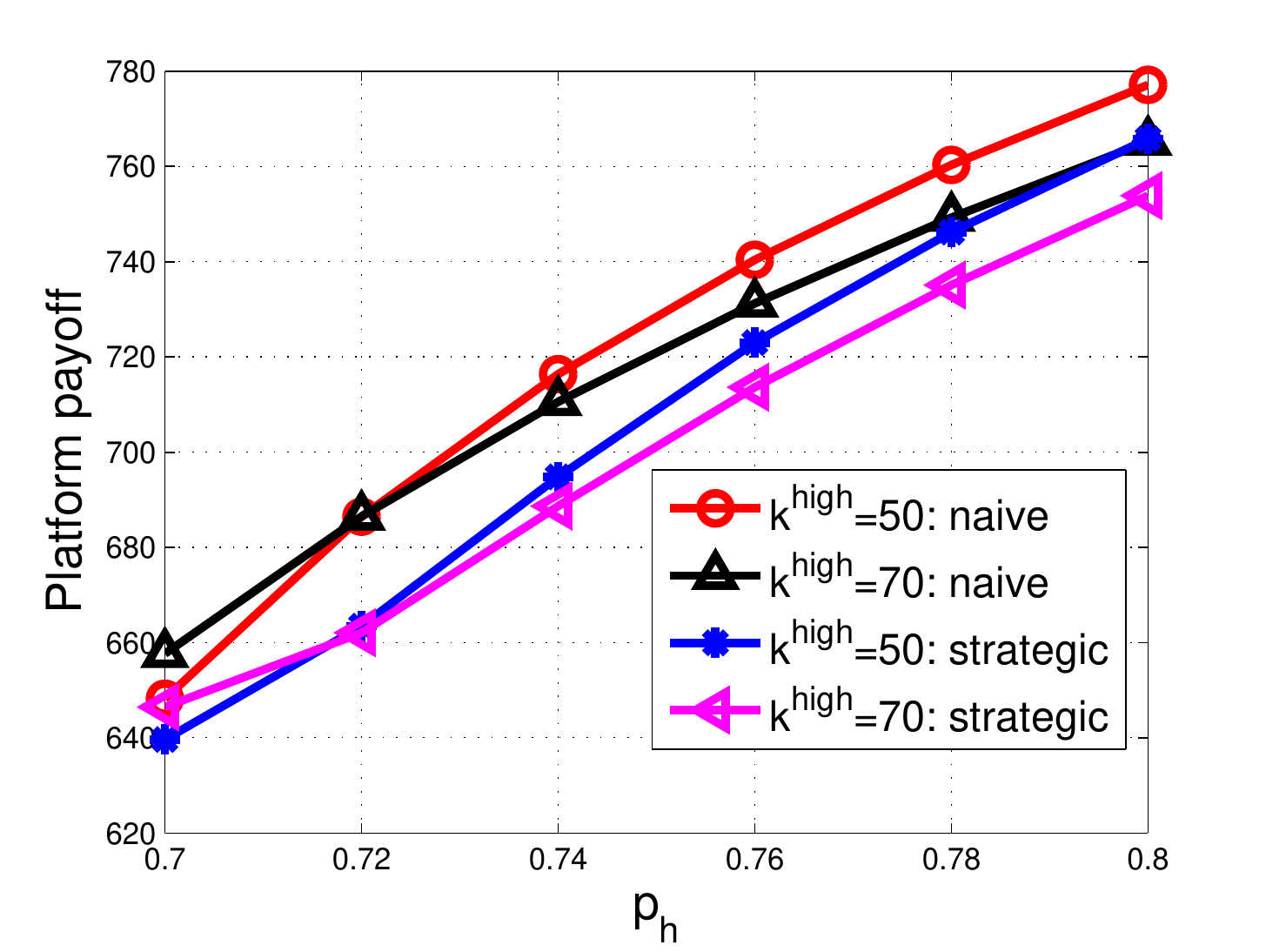}
		\label{payoffvsph}}
	\hfil
	\subfloat[Aggregate worker payoff vs. $p_h$. ]{\includegraphics[width=2.36in]{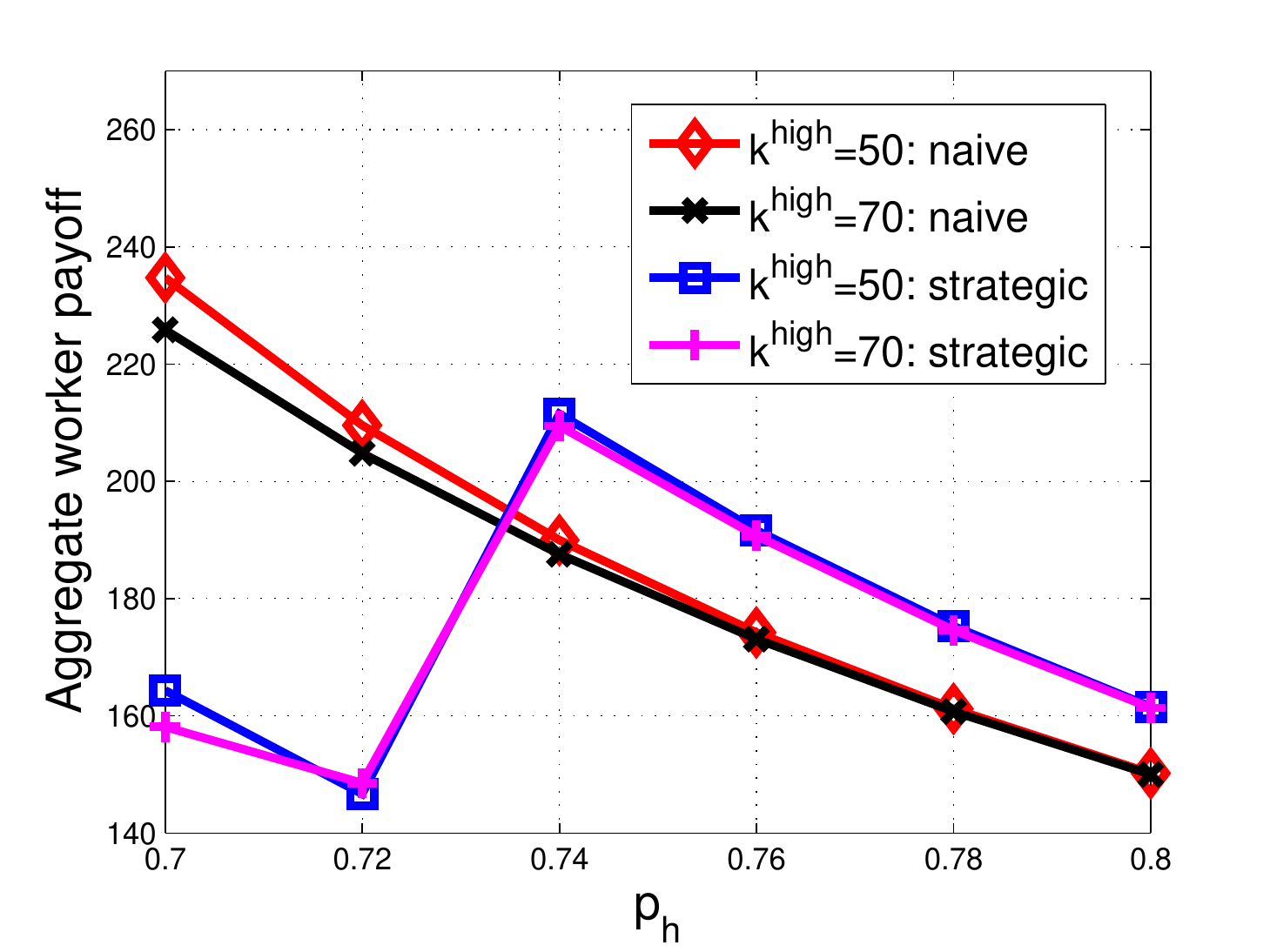}
		\label{workerpayoffvsph}}
		\hfil
		\subfloat[Social welfare vs. $p_h$. ]{\includegraphics[width=2.36in]{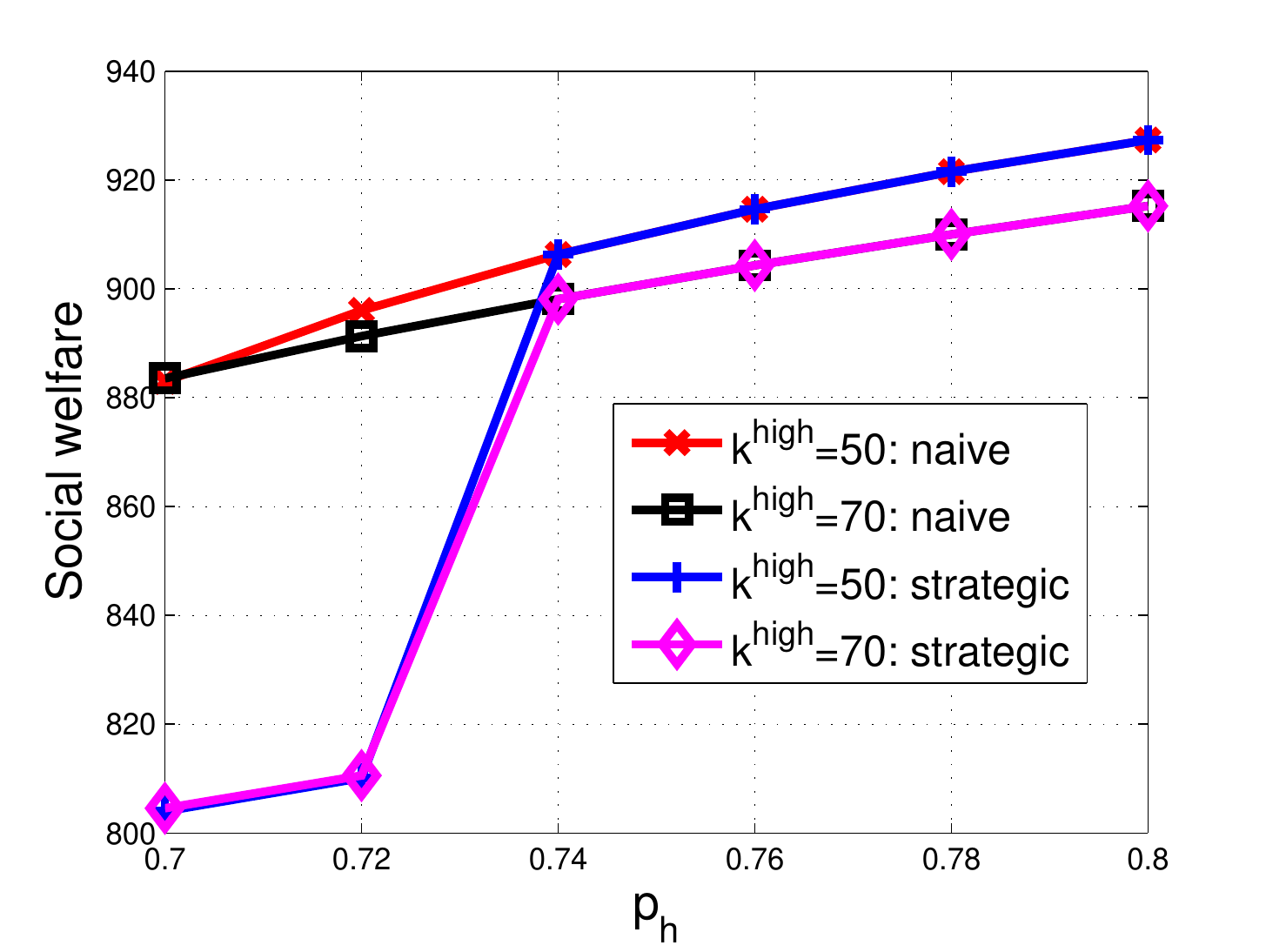}
			\label{Socialwelfarevsph}}
	\caption{Impact of worker characteristics $p_h$ on the overall mechanism performance.}
	\label{payoffandeta}
	\vspace{-5mm}
\end{figure*}
	
	We characterize a key property of the platform's optimal information revelation strategy in Theorem  \ref{info_naive}.
	\begin{The}\label{info_naive}{(Platform's Information Revelation in Stage I)}
	%(i) For naive workers, the platform's optimal information revelation strategies are $\epsilon^{h*}=1$ and $\epsilon^{l*}=0$.\\
	 It is not always optimal for the platform to set $\epsilon^h=1$ and $\epsilon^l=0$.
	\end{The}

Theorem \ref{info_naive} implies that the platform may not find it optimal to always announce $k_p^{\rm anu}=k^{\rm high}$ (see (\ref{epsilonh})) to strategic workers. This is counter-intuitive, as one may think that the platform should lull the workers into believing in a high overall worker capability (by always announcing $k^{\rm high}$) to maximize its payoff.
%\footnote{It will be shown in Section \ref{numerical}, doing so is optimal if the workers are naive.} 
%Indeed, doing so is optimal if the workers are naive (i.e., non-strategic). 
%It may not be optimal if the workers are strategic, as it may reduce the credibility of the platform's announcement and hence reduce the platform's payoff.

To understand this surprising result, we elaborate the rationale behind Theorem \ref{info_naive} as follows. Recall that the true value of $k$ affects the reward design in Stage II, and the platform's announced value $k_p^{\rm anu}$ affects the worker behaviors in Stage III. Therefore, to analyze the information revelation in Stage I, we consider the following four cases:
	\begin{itemize}
		\item \textit{Case (h,h)}: $k=k^{\rm high}$ and $k_p^{\rm anu}=k^{\rm high}$, which happens with a probability $Q_{h,h}\left(\boldsymbol{\epsilon}\right)=\mu^{\rm prior}_{\rm high}\cdot \left(1-\epsilon^l\right)$.
		\item \textit{Case (h,l)}: $k=k^{\rm high}$ but $k_p^{\rm anu}=k^{\rm low}$, which happens with a probability $Q_{h,l}\left(\boldsymbol{\epsilon}\right)=\mu^{\rm prior}_{\rm high}\cdot \epsilon^l$.
		\item \textit{Case (l,h)}: $k=k^{\rm low}$ but $k_p^{\rm anu}=k^{\rm high}$, which happens with a probability $Q_{l,h}\left(\boldsymbol{\epsilon}\right)=\mu^{\rm prior}_{\rm low}\cdot \epsilon^h$.
		\item \textit{Case (l,l)}: $k=k^{\rm low}$ and $k_p^{\rm anu}=k^{\rm low}$, which happens with a probability $Q_{l,l}\left(\boldsymbol{\epsilon}\right)=\mu^{\rm prior}_{\rm low}\cdot \left(1-\epsilon^h\right)$.
	\end{itemize}
%We can write the expected platform payoff in Stage I as:
%\begin{equation}\label{payoff_stage2}
%\begin{aligned}
%\mathbb{E}\left\{U_p\left(\boldsymbol{\epsilon}\right)\right\}&=Q_{h,h}\left(\boldsymbol{\epsilon}\right)\cdot U_{h,h}\left(\boldsymbol{\epsilon}\right)+Q_{h,l}\left(\boldsymbol{\epsilon}\right)\cdot U_{h,l}\left(\boldsymbol{\epsilon}\right)\\
%&\hspace{-2mm}+Q_{l,h}\left(\boldsymbol{\epsilon}\right)\cdot U_{l,h}\left(\boldsymbol{\epsilon}\right)+Q_{l,l}\left(\boldsymbol{\epsilon}\right)\cdot U_{l,l}\left(\boldsymbol{\epsilon}\right),
%\end{aligned}
%\end{equation}

 To solve $\boldsymbol{\epsilon}^*=\left(\epsilon^{h*}, \epsilon^{l*}\right)$ in Stage I, we first consider a fixed $\epsilon^l$, and write the expected platform payoff as follows:
\begin{equation}\label{payoff_stage3}
\begin{aligned}
\mathbb{E}\left\{U_p\left(\epsilon^h\right)\right\}&=Q_{h,h}\cdot U_{h,h}\left(\epsilon^h\right)+Q_{h,l}\cdot U_{h,l}\left(\epsilon^h\right)\\
&\hspace{-10mm}+Q_{l,h}\left(\epsilon^h\right)\cdot U_{l,h}\left(\epsilon^h\right)+Q_{l,l}\left(\epsilon^h\right)\cdot U_{l,l}\left(\epsilon^h\right),
\end{aligned}
\end{equation}
where $U_{h,h}\left(\epsilon^h\right)$ represents the maximum platform payoff under \textit{Case (h,h)} after the platform optimizes the reward in Stage II. The other notations $U_{h,l}\left(\epsilon^h\right), U_{l,h}\left(\epsilon^h\right)$, and $U_{l,l}\left(\epsilon^h\right)$ are similarly defined. Notice that  $U_{h,h}\left(\epsilon^h\right)$ depends on $\epsilon^h$, as it affects the strategic workers' posterior belief, and hence the reward design and the platform payoff.

Now, we will show that the expected platform payoff $\mathbb{E}\left\{U_p\left(\epsilon^h\right)\right\}$ may not be monotone in $\epsilon^h$. To see this, we first summarize some monotonicity results in Lemma \ref{monotone}.
\begin{Lem}\label{monotone} 
	(i) The terms $U_{h,h}\left(\epsilon^h\right)$, $U_{l,h}\left(\epsilon^h\right)$, and $Q_{l,l}\left(\epsilon^h\right)$  in (\ref{payoff_stage3}) decrease in $\epsilon^h$.\\
	(ii) The terms $U_{h,l}\left(\epsilon^h\right)$, $U_{l,l}\left(\epsilon^h\right)$, and $Q_{l,h}\left(\epsilon^h\right)$ in (\ref{payoff_stage3}) increase in $\epsilon^h$.
\end{Lem}
One can easily verify the monotonicity of $Q_{l,l}\left(\epsilon^h\right)$ and $Q_{l,h}\left(\epsilon^h\right)$ in $\epsilon^h$ based on the definitions. We focus on the monotonicity of the platform payoff under different cases. Consider $U_{h,h} \left(\epsilon^h\right)$ under \textit{Case (h,h)}, which decreases in $\epsilon^h$, as an example. 
%Choosing a larger $\epsilon^h$   increases the likelihood that the workers hear $k_{p}^{\rm anu}=k^{\rm high}$. 
%However, upon hearing $k_{p}^{\rm anu}=k^{\rm high}$, the workers will behave such that the platform payoff decreases in $\epsilon^h$ (i.e., a smaller $U_{h,h} \left(\epsilon^h\right)$). This is because a larger $\epsilon^h$ requires the platform to use larger rewards to incentivize the workers (see Corollary \ref{Reward_epsilonh}), which compromises the platform payoff. 
As shown in Corollary \ref{Reward_epsilonh}, under a larger $\epsilon^h$, the platform needs to use larger rewards to incentivize the workers, which can decrease the platform payoff. 
Similar discussions apply for $U_{l,h}\left(\epsilon^h\right)$, $U_{h,l}\left(\epsilon^h\right)$, and  $U_{l,l}\left(\epsilon^h\right)$. Because different components of $\mathbb{E}\left\{U_p\left(\epsilon^h\right)\right\}$ in (\ref{payoff_stage3}) have different monotonic properties, the expected platform payoff is not always monotonic in $\epsilon^h$ (we have conducted numerical experiments to validate the non-monotonicity). One can also show that given $\epsilon^h$, the expected payoff is not always monotonic in $\epsilon^l$. In other words, the platform even has incentives to announce $k^{\rm low}$ when $k=k^{\rm high}$ (i.e., using a positive $\epsilon^{l*}$). To conclude, the platform needs to achieve a tradeoff between being honest and lying.

The closed-form solutions of $\epsilon^{h*}$ and $\epsilon^{l*}$ are hard to derive, mainly due to the complex
coupling between the reward design and the information revelation
via (\ref{condition_psne}), (\ref{Optimal_R_case1}), (\ref{Optimal_R_case2}) (see Theorem \ref{PlatformDecision2}).  Nevertheless, besides the above analysis, we will  construct numerical examples in Section \ref{numerical} to validate Theorem \ref{info_naive}.

	\vspace{-1mm}
	\section{Numerical Results}\label{numerical}
	%Parameter Settings for figures 4a and ab
	%N=100, pl=0.6, klow=10, c=1, l=0.1, beta=1000; mu0high=0.7;
	In this section,  we provide numerical results to investigate the impact of the workers' characteristics and their prior belief on the overall mechanism performance. 	
	
	We study two types of workers: (i) \textit{strategic} workers who may doubt the platform's announced information (with complete analytical results shown in Section IV); (ii) \textit{naive} workers who put full trust in whatever the platform announces. We use notations similar to (\ref{belief_h_h})-(\ref{belief_l_l}) (except that the superscript is changed from ${\rm post, str}$ to ${\rm post, nai}$) to denote the naive workers' posterior belief. Specifically, we have 
	\begin{equation}
	\begin{aligned}
	\begin{cases}
	&\mu^{\rm post, nai}_{\rm high}|_{k^{\rm high}}=\mu^{\rm post, nai}_{\rm low}|_{k^{\rm low}}=1,\\
	&\mu^{\rm post, nai}_{\rm high}|_{k^{\rm low}}=\mu^{\rm post, nai}_{\rm low}|_{k^{\rm high}}=0,
	\end{cases}
	\end{aligned}
	\end{equation}	
	where $\mu^{\rm post, nai}_{\rm high}|_{k^{\rm low}}$ represents a naive worker's posterior belief in $k=k^{\rm high}$ conditional on the platform's announcement $k_{p}^{\rm anu}=k^{\rm low}$. A naive worker will discard his prior belief and fully believe in the announced value, e.g., $\mu^{\rm post, nai}_{\rm high}|_{k^{\rm high}}=1$. The consideration of naive workers will serve as a benchmark comparison to strategic workers.
	
	As will be shown, 
	%the platform's optimal information revelation strategies are very different when facing these two types of workers. Specifically, 
	the platform always finds it optimal to announce a high average worker accuracy to naive workers, but this is not the case to strategic workers. 
	%In fact, the platform even has incentives with strategic workers to announce a lower average accuracy than its actual value. 
	We also show a  counter-intuitive result that the platform's payoff increases in the \emph{solution accuracy} of the high-accuracy workers but may decrease in the \emph{number} of these high-accuracy workers.

	\subsection{Impact of Worker Characteristics}
	In this subsection, we study how the optimal platform payoff, the aggregate worker payoff (defined as the summation of all the workers' payoffs), and the social welfare (defined as the summation of the platform's and all the workers' payoffs) depend on the high-accuracy workers' solution accuracy $p_h$. In the experiments, we set $N=100$, $p_l=0.6$, $\mu^{\rm prior}_{\rm high}=0.7$, $\mu^{\rm prior}_{\rm low}=0.3$, $k^{\rm low}=20$,  $c=1$, $\beta=1000$, choose $k^{\rm high}$ from the set $\left\{50,70\right\}$, and change $p_h$ from $0.7$ to $0.8$ with a step size $0.02$. Fig. \ref{payoffvsph}, Fig. \ref{workerpayoffvsph}, Fig. \ref{Socialwelfarevsph} illustrate how the platform's optimal payoff, the aggregate worker payoff, and the social welfare change with $p_h$ under different $k^{\rm high}$, respectively.

		In Fig. \ref{payoffvsph}, we observe that the platform payoff increases in $p_h$. As the solution accuracy of the high-accuracy workers improves, the platform can generate the aggregated solution with  a higher accuracy and use smaller rewards to incentivize the workers. This leads to a higher platform payoff. However, we observe that given $p_h$ (e.g., $p_h=0.76$), the platform payoff may decrease in $k^{\rm high}$ (e.g., strategic workers). This is because a larger number of high-accuracy workers brings a marginally decreasing benefit to the platform, yet the total rewards may grow drastically. Note that the above observations are robust, as they hold in both the strategic and naive worker cases. One may think this conclusion resembles the one in \cite{chaoGC19}, yet they are derived under significantly different game environments. Different from \cite{chaoGC19}, we derive the results further accounting for both the information asymmetry and the strategic revelation between the platform and the workers. Next, we summarize the above observations as follows:

	\begin{figure*}[t]
		\vspace{-5mm}
		\centering
		\subfloat[Platform payoff vs. $\mu_{\rm high}^{\rm prior}$. ]{\includegraphics[width=2.36in]{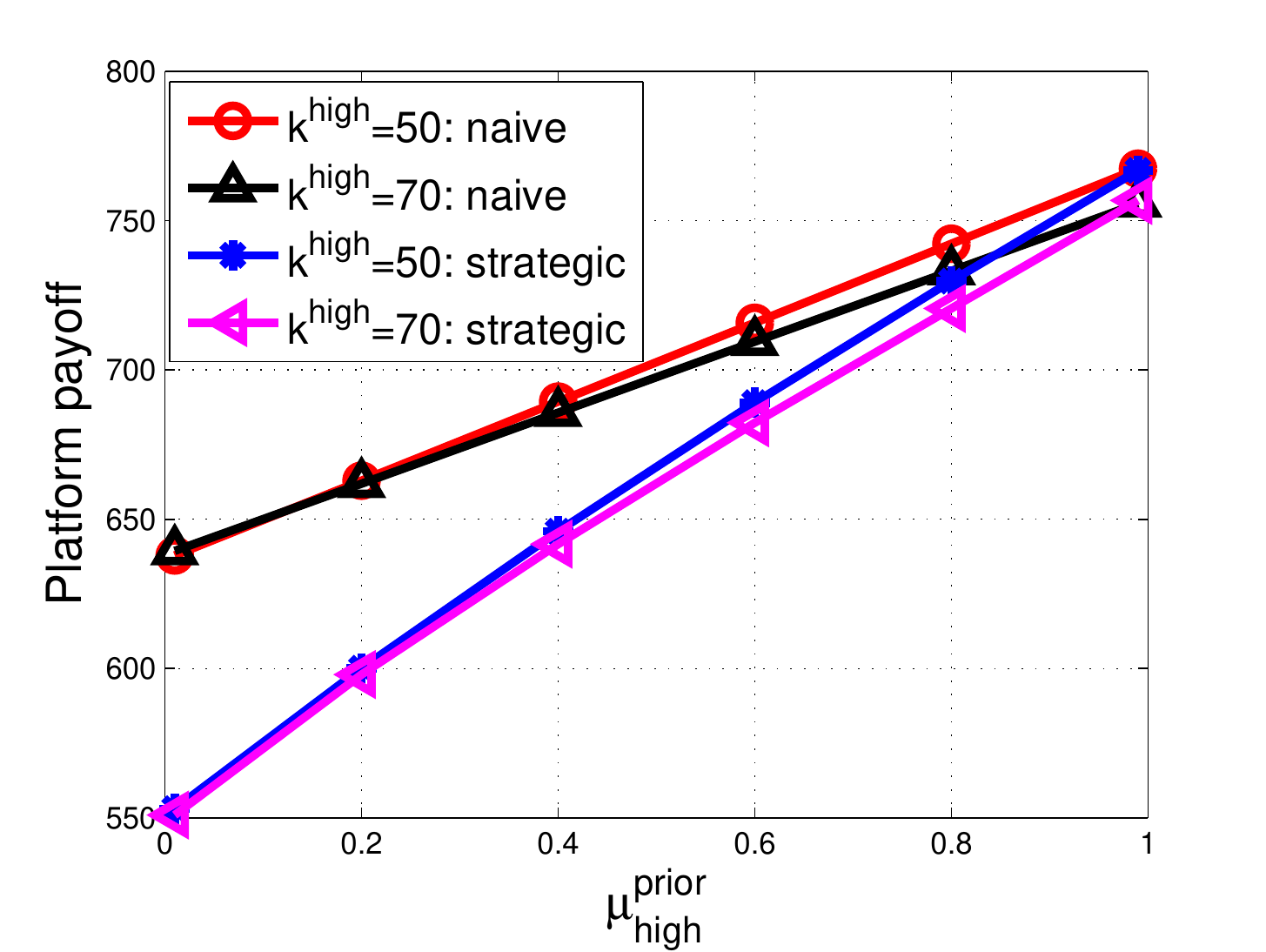}
			\label{payoffvsmu0}}
		\hfil
		\subfloat[Aggregate worker payoff vs. $\mu_{\rm high}^{\rm prior}$. ]{\includegraphics[width=2.36in]{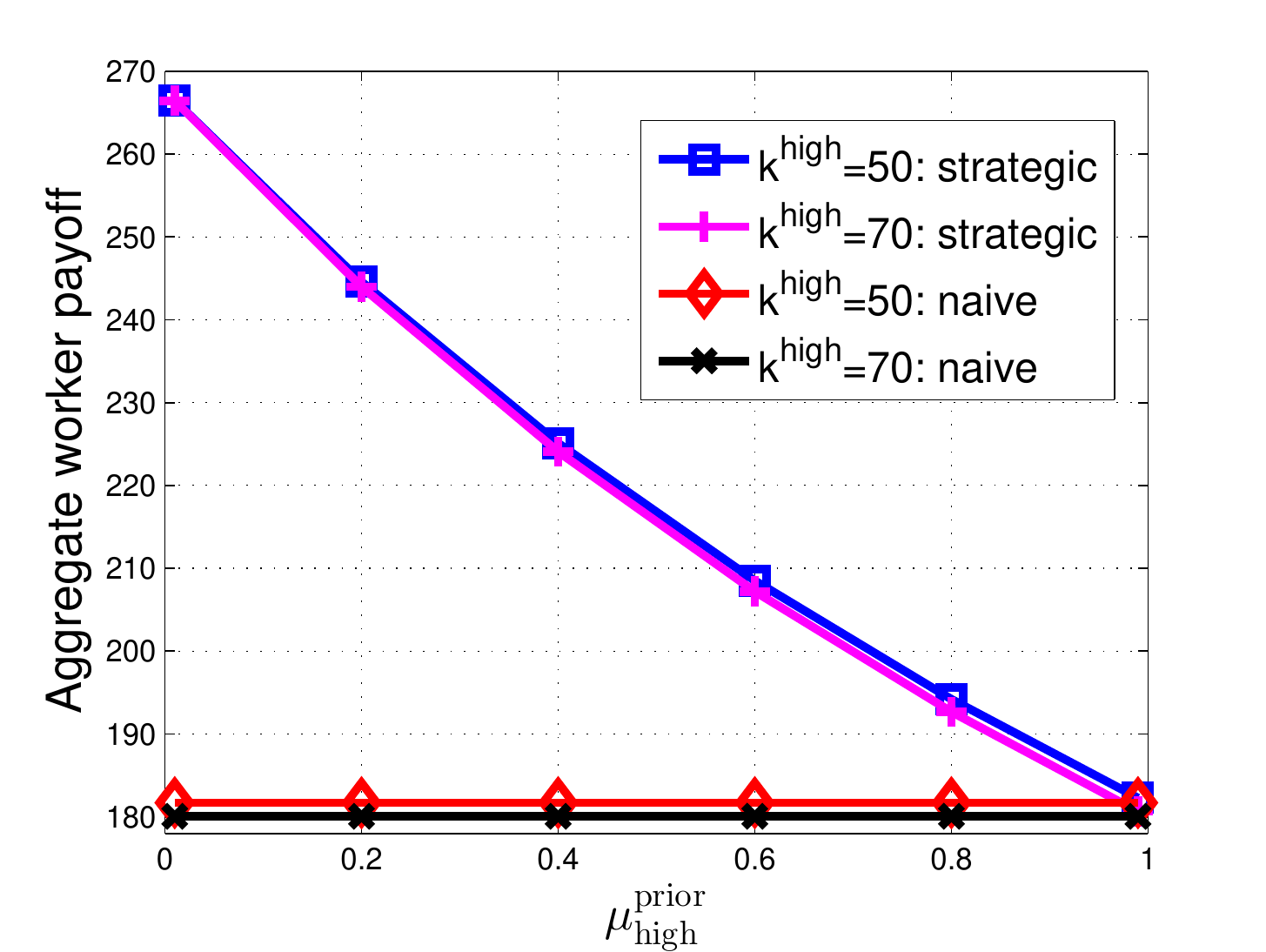}
			\label{workerpayoffvsmu0}}
		\hfil
		\subfloat[Platform information revelation vs. $\mu_{\rm high}^{\rm prior}$. ]{\includegraphics[width=2.36in]{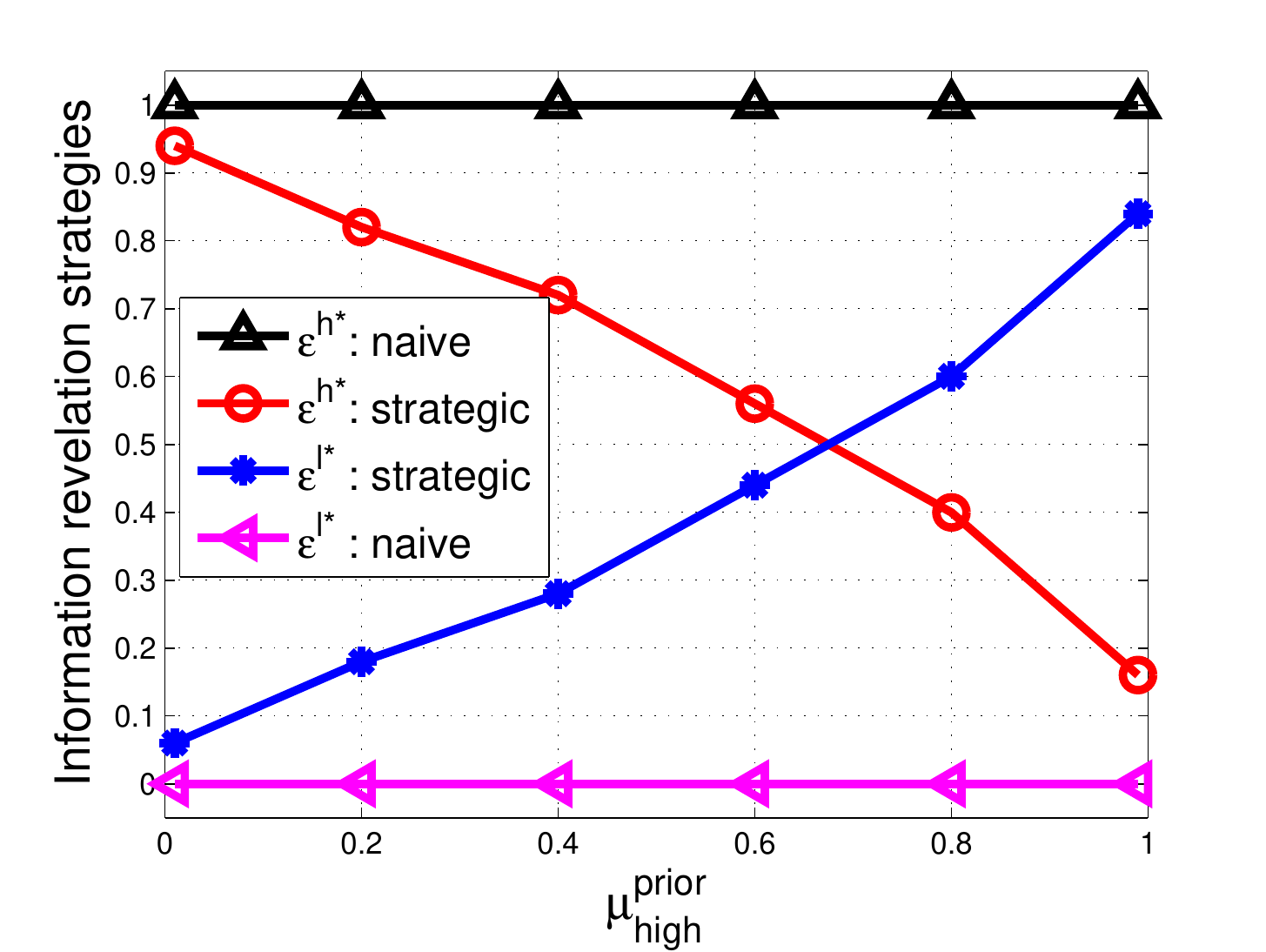}
			\label{epsilonhl}}
		
		%	\hfil
		%	\subfloat[ ]{\includegraphics[width=2.2in]{./figure/payoffvsq.eps}
		%		\label{Num_Reporter_sigmas}}
		\caption{Impact of worker prior belief $\mu_{\rm high}^{\rm prior}$ on the overall mechanism performance.}
		\label{priorbelief}
        \vspace{-3mm}
	\end{figure*}
	
	\begin{Obs}
		The platform's optimal payoff increases in the high-accuracy workers' solution accuracy $p_h$,  but it may decrease in the number of the high-accuracy workers $k^{\rm high}$.
	\end{Obs}

    In Fig. \ref{payoffvsph}, we also observe that given $p_h$ and $k^{\rm high}$ (e.g., $p_h=0.76$, $k^{\rm high}=50$), the platform payoff in the naive worker case is always larger than that in the strategic worker case. The platform can better manipulate the naive workers' belief, which helps achieve a higher platform payoff. We summarize the observation as follows:
    %As a result,  we observe in Fig. \ref{workerpayoffvsph} that, the strategic workers achieve a larger aggregate payoff than the naive workers (e.g., $k^{\rm high}=50$). Next, we summarize the above observations as follows:

	\begin{Obs}
		The platform benefits from workers' naiveness. 
		%Furthermore, the strategic workers  yield a larger aggregate payoff than that of the naive workers.
	\end{Obs}
	
	In Fig. \ref{workerpayoffvsph}, it is interesting to observe that the aggregate worker payoff may decrease in $p_h$ (e.g., $p_h\ge 0.74$ for strategic workers). This is because the platform can harvest larger benefits from more capable (accurate) workers, leading to smaller worker payoffs. This result also bears important strategic implications. Consider the scenario where a platform aims to estimate workers' accuracy by asking them to report this information \cite{chaoWiOpt20}. Workers may not reveal a high accuracy (even if it is the truth), as doing so may benefit the platform but hurt the workers. We summarize the observation as follows:
	\begin{Obs}
		More capable workers may suffer: a worker population with an overall higher accuracy may unexpectedly obtain a smaller aggregate payoff. 
	\end{Obs}

In Fig. \ref{Socialwelfarevsph}, we observe that the social welfare increases in $p_h$. Note that the reward terms are canceled in the social welfare, and a larger $p_h$ enables the platform to generate an aggregated solution with higher accuracy. As a result, the social welfare improves. We summarize the observation as follows:
\begin{Obs}
	The social welfare increases in the high-accuracy workers' solution accuracy $p_h$.
\end{Obs}

%In Fig. \ref{Socialwelfarevsph}, it is surprising to observe that the social welfare may decrease in $p_h$ (e.g., $k^{\rm high}=70$ for strategic workers).
	%It indicates that the platform benefits from the naiveness of the workers.
	 %This is because the platform can always announce $k_p^{\rm anu}=k^{\rm high}$ to misguide the naive workers into believing  $k=k^{\rm high}$, which helps reduce the rewards to incentivize them. This benefits the platform. When facing strategic workers, however, the platform needs to tackle an additional tradeoff between being honest and lying, which compromises the platform payoff. As a result, the platform payoff in the naive worker case is no smaller than that in the strategic worker case.

    \subsection{Impact of Worker Prior Belief}
    In this subsection, we study the impact of the workers' prior belief on the optimal platform payoff, the aggregate worker payoff, and the corresponding information revelation strategies.\footnote{We apply the exhaustive search algorithm to find the optimal information revelation strategies. This is reasonable as they are long term strategies and hence are fixed within a long period.} In the experiments, we set  $N=100$, $p_l=0.6$, $p_h=0.75$, $k^{\rm low}=20$,  $c=1$, $\beta=1000$, and consider $\mu^{\rm prior}_{\rm high} \in \left\{0.01, 0.2, 0.4, 0.6, 0.8, 0.99\right\}$.\footnote{We do not assign $0$ or $1$ to  $\mu^{\rm prior}_{\rm high}$ because under either assignment, the problem degenerates to the case where the platform's information revelation strategy imposes no effect on the strategic workers' posterior belief. For example, if $\mu^{\rm prior}_{\rm high}=0$, any $\epsilon^h$ and $\epsilon^l$ will result in the same posterior belief $\mu^{\rm post, str}_{\rm high}|_{k^{\rm high}}=0$ (see (\ref{belief_h_h})).}
    
    \subsubsection{Impact of Worker Prior Belief on Platform/Worker Payoff}
    We first study how the optimal platform payoff and the aggregate worker payoff are affected by the workers' prior belief, as shown in Fig. \ref{payoffvsmu0} and Fig. \ref{workerpayoffvsmu0}, respectively. 
    
    In Fig. \ref{payoffvsmu0}, we observe that given $k^{\rm high}$, the platform's optimal payoff increases in $\mu^{\rm prior}_{\rm high}$. The reasons are two-fold. First, the real value of $k$ is drawn according to  $\boldsymbol{\mu}^{\rm prior}=\left(\mu^{\rm prior}_{\rm high}, \mu^{\rm prior}_{\rm low}\right)$. A larger $\mu^{\rm prior}_{\rm high}$ implies that the number of high-accuracy workers is more likely to be $k^{\rm high}$ than $k^{\rm low}$. Hence, the platform is more likely to generate an aggregated solution with a higher accuracy. This holds for both the strategic and naive workers. Second, (i) for strategic workers, the more they are inclined to believe  $k=k^{\rm high}$ a priori, the more they also believe  $k=k^{\rm high}$  following the information announcement. Hence, the platform can use smaller rewards to incentivize them, and hence achieves a higher payoff. (ii) For naive workers, they will trust whatever information is announced and discard the prior. Hence, the rewards needed do not depend on $\mu^{\rm prior}_{\rm high}$. The platform's optimal payoff also increases due to a more accurate aggregated solution.
    
    In Fig. \ref{workerpayoffvsmu0}, we observe that the strategic workers' aggregate payoff decreases in $\mu_{\rm high}^{\rm prior}$. This is because the platform will use smaller rewards to incentivize them, which leads to a smaller aggregate worker payoff. Interestingly, the naive workers' aggregate payoff does not change in $\mu_{\rm high}^{\rm prior}$, since the rewards  do not depend on $\mu_{\rm high}^{\rm prior}$, as discussed above. Based on Fig. \ref{payoffvsmu0} and Fig. \ref{workerpayoffvsmu0}, we summarize the observations as follows:
     \begin{Obs}
     	(i) The platform's optimal payoff increases in the workers' prior belief $\mu^{\rm prior}_{\rm high}$. \\
     	(ii) The strategic workers' aggregate payoff decreases in $\mu^{\rm prior}_{\rm high}$, while that of naive workers is independent of $\mu^{\rm prior}_{\rm high}$.
     	\end{Obs}
     %It should also be noted that such an observation holds in both the naive worker case and the strategic worker case.
	
	%Second, given $k$, a larger $p_h$ yields a higher optimal payoff for the platform. For a set of fixed worker parameters (i.e., $N$, $k$, and $p_l$), a higher $p_h$ always yields a better  accuracy of the aggregated estimate and a smaller total consistency reward, leading to a higher platform payoff. 
	\subsubsection{Impact of Worker Prior Belief on Platform Information Revelation}
	%Now, we study how the platform's optimal information revelation strategies are affected by the workers' prior belief. 
	Fig. \ref{epsilonhl} illustrates how the platform's information strategies change with $\mu_{\rm high}^{\rm prior}$ (consider $k^{\rm high} =70$). We observe that $\epsilon^{h*}=1$ and $\epsilon^{l*}=0$ for naive workers. The platform should always announce a high average worker accuracy to naive workers, as it will require the minimum rewards to incentivize them.\footnote{The formal proof of this result is ready, yet we omit it in this paper due to space limitations.} However, this may not be true for strategic workers where the platform chooses $\epsilon^{h*}$  smaller than $1$ (e.g., $\mu^{\rm prior}_{\rm high}=0.4$ for the red curve). Interestingly, the platform even has incentives to announce a lower average solution accuracy than its actual value by choosing a positive $\epsilon^{l*}$ (e.g., $\mu^{\rm prior}_{\rm high}=0.6$ for the blue curve). These observations validate Theorem \ref{info_naive}, and are summarized as follows:
	 
	\begin{Obs}
		The platform always finds it optimal to announce $k^{\rm high}$ to naive workers, but not to strategic workers. It may even announce $k^{\rm low}$ to strategic workers when $k=k^{\rm high}$.
		\end{Obs}

	In Fig. \ref{epsilonhl}, when the strategic workers are more confident in $k=k^{\rm high}$ a priori (i.e., a larger $\mu_{\rm high}^{\rm prior}$), the platform should  announce $k_p^{\rm anu}=k^{\rm high}$ less frequently (i.e., a smaller $\epsilon^{h*}$ and a larger $\epsilon^{l*}$).  As a result, the strategic workers will be more inclined to believe $k=k^{\rm high}$. This benefits the platform since the required rewards to incentivize the workers can be reduced. We summarize the observations as follows:
	 
	 \begin{Obs}
	 	For strategic workers, the platform's optimal information revelation strategy $\epsilon^{h*}$ decreases in the workers' prior belief $\mu^{\rm prior}_{\rm high}$, while $\epsilon^{l*}$ increases in $\mu^{\rm prior}_{\rm high}$.
	 \end{Obs}
%	
%     \vspace{1.5mm}
%	\hspace{-3.5mm}\textbf{Observation 5.} \textit{In the strategic worker case, the platform's optimal information revelation strategies $\epsilon^{h*}$ decreases in the workers' prior belief $\mu^{\rm prior}_{\rm high}$, but $\epsilon^{l*}$ increases in $\mu^{\rm prior}_{\rm high}$.}  
%	\vspace{1.5mm}

%	\subsection{Extended Model}
%	Next, we study the extended model (see (\ref{extension})) to examine whether the platform has incentives to announce $k^{\rm low}$ when the truth is $k^{\rm high}$. Specifically, we investigate the impact of $p_h$ on the platform's information revelation strategy $\epsilon^{l*}$, given different values of $\epsilon^h$ (see Fig. \ref{regretdiffph}). 
%	
%	\textbf{Impact of worker characteristics on information revelation strategy $\epsilon^{l*}$}. In Fig. \ref{regretdiffph}, We observe that  $\epsilon^{l*}$ may decrease in $p_h$ (e.g., $\epsilon^h=0.2$). This is due to the additional tradeoff coupling with $\epsilon^h$. 
	%
	%Moreover, for a fixed $q$ (e.g., $q=0.8$), the optimal $c^*$ decreases with $\lambda$. The reason is similar to that illustrated in Fig. \ref{Opt_c_p}, i.e., a larger $\lambda$ implies a lower average cost, the platform could choose $\boldsymbol{t} $ to achieve a smaller $c^*$ to save the cost but still maintain the solution accuracy.

	\section{Conclusion}\label{conclusion}
	%\vspace{-1mm}
	In this paper, we study strategic information revelation in an IEWV problem.  The problem is a challenging non-convex program, yet we exploit its special structure to characterize the properties of the optimal solutions. We show that for naive workers, the platform should always announce  a high average worker accuracy. However, for strategic workers, it needs to tackle a tradeoff and may even have an incentive to announce an average accuracy lower than the actual value. Moreover, we show the surprising result that the platform payoff may decrease in the number of high-accuracy workers.
	
	For the future work, we plan to study the problem under multi-dimensional worker heterogeneity, where both the workers' costs and solution accuracy are heterogeneous. Moreover, it will be interesting to study costly information revelation (with cost incurred by information acquisition) in future work.

	\bibliographystyle{IEEEtran}
	% argument is your BibTeX string definitions and bibliography database(s)
	\bibliography{ref}
	%
	% <OR> manually copy in the resultant .bbl file
	% set second argument of \begin to the number of references
	% (used to reserve space for the reference number labels box)
	%\begin{thebibliography}{1}
	%
	%%\bibitem{IEEEhowto:kopka}
	%%H.~Kopka and P.~W. Daly, \emph{A Guide to \LaTeX}, 3rd~ed.\hskip 1em plus
	%%  0.5em minus 0.4em\relax Harlow, England: Addison-Wesley, 1999.
	%%  
	%
	%
	%\end{thebibliography}

	% that's all folks

\end{document}